\documentclass{article}
\usepackage[in]{fullpage}
\usepackage{amsmath,amssymb,amsthm,mathtools,mathrsfs,bm}
\usepackage{helvet}
\usepackage{courier}
\usepackage[hyphens]{url}
\usepackage{graphicx}
\usepackage[square,numbers]{natbib}
\usepackage{algorithm}
\usepackage[noend]{algorithmic}

\title{Online Task Assignment Problems with Reusable Resources}
\author{
Hanna Sumita\thanks{Tokyo Institute of Technology, sumita@c.titech.ac.jp} \and
Shinji Ito\thanks{NEC Corporation, \{kei\_takemura, i-shinji\}@nec.com} \and
Kei Takemura\textsuperscript{$\dag$} \and
Daisuke Hatano\thanks{RIKEN AIP, daisuke.hatano@riken.jp} \and
Takuro Fukunaga\thanks{Chuo University, fukunaga.07s@chuo-u.ac.jp} \and
Naonori Kakimura\thanks{Keio University, kakimura@math.keio.ac.jp} \and
Ken-ichi Kawarabayashi\thanks{National Institute of Informatics, k\_keniti@nii.ac.jp}
}
\date{}

\usepackage{amsmath,amsthm,amssymb,mathtools,bm}
\usepackage{multirow}
\usepackage{comment}
\usepackage{subcaption}
\usepackage{refcount}

\newcommand{\veryfast}{$<\varepsilon$}
\newcommand{\Random}{R}
\newcommand{\Greedy}{G}
\newcommand{\NAdap}{N}
\newcommand{\Proposed}{Pro}
\usepackage{todonotes}
\newtheorem{theorem}{Theorem}

\newtheorem{example}{Example}
\newtheorem{lemma}{Lemma}

\newtheorem{remark}{Remark}
\newtheorem{claim}{Claim}

\newcommand{\OPT}{\mathrm{OPT}}
\newcommand{\ALG}{\mathrm{ALG}}
\newcommand{\cI}{\mathcal{I}}
\newcommand{\ang}[1]{{#1}^{\circ}}

\usepackage{enumerate}
\usepackage{enumitem}
\usepackage{booktabs}

\begin{document}

\maketitle

\begin{abstract}
We study online task assignment problem with reusable resources, motivated by practical applications such as ridesharing, crowdsourcing and job hiring.
In the problem, we are given a set of offline vertices~(agents), and, at each time, an online vertex~(task) arrives randomly according to a known time-dependent distribution.
Upon arrival, we assign the task to agents immediately and irrevocably.
The goal of the problem is to maximize the expected total profit produced by completed tasks.
The key features of our problem are
(1) an agent is reusable, i.e., an agent comes back to the market after completing the assigned task, (2) an agent may reject the assigned task to stay the market, and (3) a task may accommodate multiple agents.
The setting generalizes that of existing work in which an online task is assigned to one agent under~(1).

In this paper, we propose an online algorithm that is $1/2$-competitive for the above setting, which is tight.
Moreover, when each agent can reject assigned tasks at most $\Delta$ times, the algorithm is shown to have the competitive ratio $\Delta/(3\Delta-1)\geq 1/3$.
We also evaluate our proposed algorithm with numerical experiments.
\end{abstract}

\section{Introduction}
Online task assignment problem has attracted extensive attention recently in combinatorial optimization and artificial intelligence.
The problem models practical situations that assign agents to tasks arriving online.
For example, in a rideshare platform~\cite{Dickerson2018, Dong2021,Lowalekar2020,Nanda2020} such as Uber and Lyft, we match drivers to riders where requests from riders arrive one by one. 
Other applications include crowdsourcing~\cite{Assadi2015,Ho2012,Xu2017} and job hiring~\cite{Anagnostopoulos2018,Dickerson2019}.

In this paper, we study online task assignment problem under \emph{known adversarial distributions}.
In the problem, we are given a time horizon $T$ and a bipartite graph $G=(U,V; E)$, where $U$ corresponds to a set of agents and $V$ represents \emph{types} of tasks.
A type of a task is usually defined by attributes of the task, e.g., pick-up/drop-off locations of a rider in a rideshare platform and language skills required to complete a task in a crowdsourcing platform.
An edge $(u,v) \in E$ means that $u$ has ability to process an online task $v$.
In the problem, the set $U$ of agents is known in advance~(offline vertices), while we are given a vertex with a type $v \in V$~(an online vertex) at each time, randomly according to a time-dependent distribution $D_t$ over $V$.
We identify an online vertex with its type $v$.
On arrival of $v$, we immediately and irrevocably either assign some agents to $v$, or discard the chance of processing $v$.
Then we obtain a profit produced by an assigned agent if she completes a task.
The goal is to design an online algorithm that maximizes the total profit.
We here assume that $D_t$ is known in advance, which is called a \emph{Known Adversarial Distribution (KAD) model}~\cite{Dickerson2018}; see also~\cite{Alaei2012}.
If $D_t$ is the same for all $t$, it is called a \emph{Known I.I.D.~(KIID) model}.

Motivated by practical applications, it has been studied to generalize the online task assignment problem with additional specific conditions.
An example is \emph{reusability} of agents~\cite{Dickerson2018,Dong2021,Gong2021,Goyal2021,Rusmevichientong2020}.
In the standard online task assignment problem, when an agent is assigned a task, she will leave immediately from the market.
However, if agents are reusable, an agent comes back to the market after completing the task, and she can be matched to a new task again.
Such situation especially arises in a rideshare platform with drivers and riders.
Another example is \emph{rejections} by agents~\cite{Nanda2020}.
It is natural that an agent is allowed to reject an assigned task if it is not satisfactory.
Moreover, each agent may have an upper bound on the number of rejections, that is, when she rejects an assigned task $\Delta$ times, then she must leave from the market.
It should be noted that the previous work mentioned above studied the \emph{online bipartite matching} where only one agent can be assigned to an online task, and it is not straightforward to extend\footnotemark\ the results to the online task assignment problem.\footnotetext{We can consider a simple reduction that replaces an online vertex $v$ with capacity $b_v$ as $b_v$ vertices with capacity $1$ arriving sequentially. Such reduction, however, would not work in the KIID/KAD model because an online vertex is chosen independently at each time.}

In this paper, we introduce the online task assignment problem in the KAD model with all the constraints mentioned above, that is, the problem with the following constraints: 
\begin{enumerate}[label=\Alph*.]
    \item \textbf{(Reusability)} an agent comes back to the market after she completes a task, where the occupation time is drawn from a known distribution,
    \item \textbf{(Rejections)} an agent $u$ rejects an assignment with some probability, and has to leave after rejecting $\Delta_u$ times, and
    \item \textbf{(Task capacity)} an online vertex $v$ can accommodate at most $b_v$ offline vertices.
\end{enumerate}
See Section~\ref{sec:model} for a formal problem definition.

\subsection{Our Contribution}
Our main result is to propose an online algorithm for the above problem.
We prove that our algorithm has the following theoretical guarantees:
\begin{itemize}
    \item $1/2$-competitive for the unlimited rejection case, i.e., when $\Delta_u$ is $+\infty$ for all $u$~(Theorem~\ref{thm:main-unlimited}),
    \item $\Delta/(3\Delta-1)$-competitive for the general case (Theorem~\ref{thm:main-general}), where $\Delta = \max_{u\in U}\Delta_u$. 
\end{itemize}
We here evaluate the performance of online algorithms with the common measure called \emph{competitive ratio}.
We say that an algorithm is $\alpha$-competitive if the expected profit obtained by the algorithm is at least $\alpha$ times the offline optimal value~(See Section~\ref{sec:compratio}).

Our results and existing results are summarized in Table~\ref{tab:summary}.
Note that our algorithm is useful for the problem without reusability.
In fact, we prove that our algorithm has the competitive ratio $(1-1/e^2)/2>0.432$ in the KIID model without reusability, which improves on the competitive ratio $1/e$~\cite{Nanda2020} for the same setting~(Theorem~\ref{thm:KIID}).

In addition, we show that no (adaptive) algorithm can achieve the competitive ratio better than $1/2$.
This is implied by a well-known example for the prophet inequality problem~\cite{krengel1977semiamarts,krengel1978semiamarts}. 
We note that this hardness result is applied even 
in the setting of~\cite{Dickerson2018}.
Thus our result complements their result, which considers only non-adaptive algorithms.

Below let us describe technical highlights of the proof of our algorithms.

\begin{table*}[htb]
    \centering
    \caption{Summary of our results and previous work. The algorithm~\cite{Dickerson2018} requires to know in advance the probability that an agent is available at each time.} 
    \label{tab:summary}
    \scalebox{0.95}{
    \begin{tabular}{c|cccccc}
    \toprule
     & Reusability & \# Rejections & \begin{tabular}{c}\# Assigned \\ agents\end{tabular} &  \begin{tabular}{c}Online \\ vertices\end{tabular} & \begin{tabular}{c}Competitive\\ ratio\end{tabular} & Hardness \\ \hline
        \textbf{This work} & \checkmark & $+\infty$ & $\geq 1$ & KAD & $1/2$ & \multirow{2}{*}{$1/2$}\\
        \textbf{This work} & \checkmark & $\leq \Delta$ & $\geq 1$ & KAD & $\Delta/(3\Delta-1)$
        & \\
        \citet{Dickerson2018} & \checkmark & NA & $1$ & KAD &$1/2$ & ($1/2$) \\ \hline
        \textbf{This work} & & $\leq \Delta$ & $\geq 1$ & KIID & $>0.432$ & \\
        \citet{Nanda2020} & & $\leq \Delta$ & $1$ & KIID & $1/e$ & $1-1/e$\\ 
        \bottomrule
    \end{tabular}
    }
\end{table*}

\paragraph{Technique}

To design an online algorithm, we first construct an offline linear program~(LP) that gives an upper bound on the offline optimal value.
This is a similar approach to the online bipartite matching in the KIID/KAD model~\cite{Alaei2012,Dickerson2018,Nanda2020}.
Our LP has variables $x_{uv, t}$ for each $(u, v)\in E$ and $t$, and has linear inequalities that incorporate the constraints (A)--(C) above. 
See Section~\ref{sec:model} for the formal definition of our LP.
Intuitively, each variable $x_{uv, t}$ corresponds to a probability that $v$ arrives at time $t$ and $u$ is assigned to $v$.
We shall use an LP optimal solution $x^*_{uv,t}$ to determine how to assign agents to $v$ at each time.
There are, however, several difficulties to obtain our results as described below.

The first difficulty is to handle the task capacity condition, i.e., that we choose a set of agents each time.
Given LP optimal solution $x^*_{e,t}$ for each edge $e$, we would like to have a distribution of feasible sets so that the probability of choosing $e$ at time $t$ is $x^*_{e, t}$.
For the online bipartite matching, this is easy; we just choose an edge $e$ with probability proportional to $x^*_{e, t}$ at time $t$.
However, for the online task assignment problem, such independent choice of edges may violate the task capacity condition.
To avoid it, we use Carath\'eodory's theorem in convex analysis.
The theorem allows us to decompose $x^*_{e, t}$ to a polynomial number of feasible sets, which can be regarded as a probability distribution on feasible sets.

Another difficulty is that the LP optimal value may give a loose upper bound on the offline optimal value.
In this case, finding a feasible set based on LP is not sufficient; we may select an agent which should not be chosen.
In fact, there exists a problem instance such that an algorithm similarly to \citet{Nanda2020} that just assigns an agent according to an LP optimal solution 
would fail to obtain the competitive ratio more than $1/3$~(See Appendix~\ref{sec:simple-is-bad} for a specific example).
To overcome this difficulty, we adapt the idea of \citet{Alaei2012}, who analyzed a different online matching problem in the KAD model.
Specifically, after finding a feasible set $S$ based on $x^*_{e, t}$, we decide whether to assign a task $v$ to $u$ for each agent $u\in S$.
We compute the expected profit that $u$ earns at and after the current time $t$ by assigning $u$ to $v$, and, if it is larger than the one when not assigning, then we decide to assign $u$ to $v$.

To prove that the proposed algorithm admits desired competitive ratio, we evaluate the expected profit that each vertex $u$ earns.
Let $R^d_{u,t}$ be the expected profit that $u$ earns at and after time $t$, when $u$ has a remaining budget $d$.
Then 
$R^d_{u,t}$ can be represented recursively with respect to $t$ and $d$.
Since the recursive equation is linear, we can bound $R^d_{u, t}$ by introducing another linear program and its dual.
We note that when $\Delta_u$ is infinite, the recursive equation becomes simpler depending on only $t$, that gives a better competitive ratio.

Let us describe the differences from~\cite{Dickerson2018,Nanda2020}.
The algorithm of \citet{Dickerson2018} assumes that the algorithm can access a probability that $u\in U$ is \emph{available} at time $t$, i.e., $u$ is not occupied by a task at time $t$.
This assumption may be problematic because it is not easy to obtain such information precisely.
In contrast, we adopt expected profit after $t$ to decide the assignment, which can be computed deterministically and efficiently by dynamic programming.
Our algorithm is also different from \citet{Nanda2020}, as their algorithm just uses $x^*_{e, t}$ as a probability mentioned above, which works only for the KIID model.
Our algorithm is shown to admit a better competitive ratio in their setting, which improves on the competitive ratio $1/e$ by~\citet{Nanda2020}~(Theorem~\ref{thm:KIID}).

\paragraph{Experiments}
We evaluate the performance of our algorithm through experiments.
We use a synthetic dataset and the real-world dataset of taxi trip records, similarly to previous work~\cite{Dickerson2018,Nanda2020}.
Our algorithm performs the best in most cases, and runs practically fast enough.
The results imply the superiority of our algorithm.

\subsection{Related Work}

The online task assignment problem is a generalization of the online bipartite matching~(where each vertex can match to at most one neighbor).
Nowadays there exist a large body of literature on online matching.
We here mention some of them closely related to our problem.
See~\cite{Mehta2013Survey} for the detailed survey.

The online bipartite matching was introduced by \citet{Karp1990}.
They considered the adversarial input model, that is, the model that online vertices arrive in an adversarial order, and proposed a randomized $(1-1/e)$-approximation algorithm for the unweighted case.
It is known that the ratio is tight~\cite{Birnbaum2008,Mehta2007}.
When online vertices arrive independently according to a distribution~(i.e., the KIID model),
\citet{Manshadi2012} proposed a $0.702$-competitive algorithm and showed that the ratio cannot be better than $0.823$.
For the edge-weighted case, a $0.667$-competitive algorithm is known~\cite{Haeupler2011}. 
It is also known that we can do better if arrival rates (the expected number of arrivals in $T$ rounds) are integral~\cite{Brubach2016}.

\emph{Online stochastic matching}, introduced by~\citet{Mehta2012}, is the problem that an offline vertex accepts an assignment with a probability.
This can be viewed as that infinite number of rejections are allowed in the process. 
For the problem in the adversarial input model, \citet{Mehta2015} presented a $0.534$-competitive algorithm for the unweighted case when edge probabilities go to $0$.
For the KIID model, \citet{Brubach2016} gave a $(1-1/e)$-competitive algorithm, which works also in the edge-weighted case. 

Online bipartite matching in the KAD model was introduced by \citet{Alaei2012} under a name of \emph{prophet-inequality} matching, as the problem includes the prophet inequality problem~\cite{krengel1977semiamarts,krengel1978semiamarts} as a special case.
\citet{Alaei2012} extended the problem so that an offline vertex has a capacity.
When offline vertices have capacities, 
the online bipartite matching is often called the AdWords problem~\cite{Mehta2007}. 
This variant is also studied extensively~\cite{Devanur2009,Lowalekar2020}.

Recently, online task assignment with the reusability condition receives attention.
\citet{Dickerson2018} gave a $1/2$-competitive algorithm for the KAD model.
Variant problems have been studied~\cite{Dong2021,Gong2021,Goyal2021,Rusmevichientong2020}; see also references therein.
Our problem can be viewed as a general framework that unifies the problem of \citeauthor{Dickerson2018} and online stochastic matching with limited number of rejections. 
We also mention that our setting of the task capacity is close to the online assortment optimization (see e.g.~\cite{Gong2021}). In the problem, an online vertex is offered a set of offline vertices (assortment), and the online vertex selects one of them.

\citet{Nanda2020} discussed the online bipartite matching problem that maximizes fairness, instead of the total profit.
Here, the fairness means the smallest acceptance ratio in tasks.
Note that our results can easily be extended to their setting to balance a trade-off between profit and fairness.

\section{Model}\label{sec:model}
In this section, we describe a formal definition of our problem. 
For a positive integer $k$, we denote $[k]=\{1,\dots, k\}$.
We are given a bipartite graph $G=(U,V; E)$ with edge weight $w_e \geq 0$, where $U$ is the set of offline vertices and $V$ is the set of types of online vertices. 
We are also given a time horizon $T$.
Each offline vertex $u \in U$ has a positive integer $\Delta_u$, which is a budget of allowed rejections in the process.
In addition, each edge $(u, v)\in E$ has a random variable $C_e\in [T]$ that represents the occupation time for $u$ to complete the task $v$.
In other words, when $u$ accepts $v$, $u$ is absent from the market, and will be available at time $t+C_{e}$, where $C_e$ is drawn from a given distribution.
We say that an offline vertex $u$ is \emph{available} at time $t$ if $u$ is in the market~(i.e., not being occupied by some task) and $u$ has not rejected online vertices $\Delta_u$ times.

For each time $t\in [T]$, an online vertex $v$ arrives according to a probability distribution $\{p_{v,t}\}_v$.\footnote{Nothing arrives at time $t$ with probability $1-\sum_{v\in V} p_{v,t}$.}
Upon arrival of a vertex $v$, we immediately and irrevocably either assign at most $b_v$ neighbors to $v$ that are available, or discard $v$. 
When $u \in U$ is assigned, $u$ either accepts $v$ with probability $q_e$ or rejects $v$ with probability $1-q_{e}$, where $e=(u,v)$. 
When $u$ accepts $v$, we obtain a profit $w_{e}$ and $u$ becomes absent from the market during the occupation time $C_e$.

We note that we may assume $\Delta_u > 0$ for all $u\in U$.
When no rejection is allowed for $u$, i.e., $\Delta_u = 0$, only edges $e$ with $q_e = 1$ can be matched to $u$.
Therefore, we can remove all the edges $e$ such that $e$ is incident to $u$ and $q_e<1$, and set $\Delta_u$ to $1$.

Our goal is to design an online (randomized) algorithm that maximizes the expected total profit. 
The performance of an online algorithm is evaluated by the competitive ratio. 
In the subsequent sections, we define the offline optimal value and the competitive ratio.

We note that our model includes the online bipartite matching problem studied in~\cite{Dickerson2018,Nanda2020} as special cases.
We assume that $b_v=1$ for each $v\in V$.
When $\Delta_u=+\infty$ for any $u\in U$ and the acceptance probability $q_e$ is one for each $e\in E$, we can ignore the constraint on rejections, and hence our model is identical to the one in~\cite{Dickerson2018}.
On the other hand, when probability distribution $\{p_{v,t}\}_v$ is the same for all $t\in [T]$ and $\Pr[C_e=T] = 1$ for any $e\in E$, our model is the exactly one in the KIID model without reusable resources, which was studied in~\cite{Nanda2020}. 
We also note that our model and the problem of \citet{Alaei2012} have no inclusion relationships.

\subsection{Offline Optimal Algorithms}
Given a problem instance, a sequence of online vertices is determined according to probability distributions $\{p_{v,t}\}_{v, t}$.
We denote by $\cI$ the probability distribution over all input sequences of online vertices. 
In the offline setting, 
we suppose that we know the sequence $I$ of online vertices in advance. 
An offline algorithm that maximizes the expected profit is called an \emph{offline optimal algorithm for $I$}.
We note that the algorithm does not know whether each offline vertex accepts or rejects an online vertex at each time. 
We denote the expected profit of the offline optimal algorithm for $I$ by $\OPT(I)$. 
Define the \emph{offline optimal value} by $\mathbb{E}_{I\sim \cI}[\OPT(I)]$.
Thus we consider the best algorithm for each $I$ in the offline optimal value.

We denote the set of edges incident to $u$ by $E_u=\{(u,v) \mid (u,v)\in E \}$ for each $u\in U$, and similarly $E_v$ for $v\in V$.
We introduce an offline LP whose optimal value is an upper bound of the offline optimal value:
		\begin{align}
		    {(\texttt{Off}) } \quad 
			\text{max} \ & \sum_{t\in [T]} \sum_{e \in E} w_e q_e x_{e,t}  \nonumber \\ 
			\text{s.t.} \ &\sum_{t'< t} \sum_{e\in E_u} x_{e,t'} q_e \Pr[C_e \geq t-t'+1] + \sum_{e\in E_u} x_{e,t} q_e \leq 1 \quad ( u\in U, t\in [T]) \label{eq:offlineLP-1}\\
			& \sum_{t\in [T]}\sum_{e\in E_u}x_{e,t}\left(1-q_e\Pr[C_e\leq T-t ]\right) \leq \Delta_u \quad (u \in U)\label{eq:offlineLP-2} \\
			& \sum_{e\in E_v} x_{e,t} \leq p_{v,t}b_v \quad (v\in V, t\in [T]) \label{eq:offlineLP-3} \\
			& 0\leq x_{e,t}\leq p_{v,t} \quad (v\in V, e\in E_v, t\in [T]) \label{eq:offlineLP-4}
		\end{align}
For each edge $e=(u,v)$ and time $t$, a variable $x_{e,t}$ corresponds to a probability that $e$ is chosen at time $t$.
Intuitively, constraints~\eqref{eq:offlineLP-1}--\eqref{eq:offlineLP-3} corresponds to the constraints (A)--(C) mentioned in the introduction, respectively.

\begin{lemma}\label{lem:offlineOPT}
The optimal value of LP {\rm (\texttt{Off})} is at least the offline optimal value $\mathbb{E}_{I\sim \cI}[\OPT(I)]$.
\end{lemma}
\begin{proof}
For each input sequence $I$, an offline optimal algorithm for $I$ determines a probability that, when an online vertex $v$ arrives at time $t$, $v$ is assigned to offline vertices $u$ for each $e=(u,v)\in E$ and $t\in [T]$.
By taking expectations over $\cI$, 
we set $x_{e,t}$ to be the probability that an online vertex $v$ arrives at $t$ and an offline vertex $u$ is assigned to $v$ for each $e=(u,v)\in E$ and $t\in [T]$.
We show that $x$ defined above is a feasible solution to the LP. 
By definition, $x_{e,t}$ is nonnegative and at most the probability $p_{v,t}$, and hence \eqref{eq:offlineLP-4} is satisfied.

Let us see that the first constraint \eqref{eq:offlineLP-1} is valid.
Fix an input sequence $I$, $u \in U$, and $t\in [T]$. 
Let $v^{t'}$ be an online vertex arriving at time $t'<t$ (if exists).
For any realization of acceptances/rejections and $C_e$'s, we have any one of the following: $u$ has no assignment, 
or $u$ is occupied by $v^{t'}$ for some $t'<t$ with occupation time at least $t-t'+1$. 
Since acceptances/rejections and $C_e$'s are sampled independently, the expected number of online vertices occupying $u$ is $\Pr[u \text{ accepts }v^t \mid I]+ \sum_{t' <t} \Pr[u \text{ accepts } v^{t'} \text{ at }t' \mid I]\Pr[C_{(u,v^{t'})} \geq t-t'+1]$.
Since $\Pr[u \text{ accepts }v^t \mid I]=x_{(u,v)}q_{(u,v)}$ and the expected number is at most $1$, \eqref{eq:offlineLP-1} holds.

To see the second constraint \eqref{eq:offlineLP-2}, let us fix an input sequence $I$ and $u \in U$.
For any realization of acceptances/rejections and $C_e$'s, 
the number of rejections plus assignments with $u$ never coming back (i.e., $C_e$ being at least $T-t+1$) is at most $\Delta_u$. 
Then the expected number of rejections is $\sum_{t\in [T]} \Pr [u\text{ rejects some } v \in V \text{ at }t\mid I] = \sum_{e\in E_u} x_{e,t}(1-q_{e})$. 
The number of assignments that $u$ never returns is $\sum_{t\in [T]} \Pr[u \text{ never comes back}\mid I] =\sum_{t \in [T]} \sum_{e\in E_u} x_{e,t}q_e \Pr[C_e \geq T-t+1]$.
Hence, we have
$\sum_{t \in [T]} \sum_{e\in E_u} (x_{e,t}(1-q_{e}) + x_{e,t}q_e \Pr[C_e \geq T-t+1]) \leq \Delta_u$, and 
\eqref{eq:offlineLP-2} holds. 

The third constraint \eqref{eq:offlineLP-3} is satisfied because the expected number of assignments made for an online vertex $v$ is at most $p_{v,t}b_v$.
Therefore, the solution $x$ defined from offline optimal algorithms is feasible to the LP.
\end{proof}

Note that the LP (\texttt{Off}) is a non-trivial extension of those in the previous papers~\cite{Dickerson2018,Nanda2020}. 
The constraints on rejections can be naturally expressed as $\sum_{t\in [T]} \sum_{e\in E_u} x_{e,t}(1-q_{e}) \leq \Delta_u$, similarly to the offline LP in \cite{Nanda2020}. 
However, this is not enough to show our result, and we need a stronger formulation \eqref{eq:offlineLP-2} of the constraints.

\begin{example}
	Let us consider the following instance. 
	Let $U=\{u\}$, $V=\{v_1,v_2,v_3\}$ and $E=\{(u,v)\mid v\in V\}$. 
	Let also $T=3$, $\Delta=1$ and $\varepsilon\ll1$.
	Assume that, at time $t$, only a vertex $v_t$ can arrive, whose probability is $p_t$. 
	For our notational convenience, we identify an edge $(u,v_t)$ with $t$. 
	Let $p_1=p_2=1$, $p_3=\varepsilon$, $q_1=q_2=1/2$, $q_3=1$, $w_1=4/9$, $w_2=6/9$, and $w_3 = 4/(9\varepsilon)$.
	We set $\Pr[C_t=1] = 1/2$, $\Pr[C_t=2]=1/2$, $\Pr[C_t= 3] = 0$ for all $t$.

		We demonstrate the calculation of the expected profit for an adaptive online algorithm. 
	    Consider an adaptive algorithm that we always assign an arriving vertex to $u$ if possible. 
        Then, at time $1$, we obtain a profit $\frac{4}{9}\cdot \frac{1}{2}=\frac{2}{9}$ in expectation.
        The probability that $u$ is available at time $2$ is $q_1 \Pr[C_t=1]=\frac{1}{4}$ because it is equal to the probability that $u$ accepts $v_1$ and $C_1$ takes $1$. 
        Thus the probability that $v_2$ is assigned to $u$ at time $2$ is 
        $\Pr[(u \text{\ is available at time 2}) \wedge (v_2\text{\ arrives at time\ } 2)] = \frac{1}{4}$.
        Hence the expected profit at time $2$ is $\frac{6}{9} \cdot \frac{1}{4} \cdot \frac{1}{2}=\frac{1}{12}$.
        By a similar discussion, the expected profit at time $3$ is $\frac{5}{36}$.
        Therefore, the expected profit in total is $\frac{2}{9}  +  \frac{1}{12} + \frac{5}{36} = \frac{4}{9}$.

	    The offline optimal value can be calculated as follows.
	    When we know $v_3$ arrives at time $3$ in advance, 
	    the best strategy is that we discard $v_1$ and $v_2$ but assign $v_3$ to obtain a profit $\frac{4}{9\varepsilon}$. 
	    On the other hand, suppose that $v_3$ does not arrive and we know it.
	    We describe that the expected profit when we assign $v_1$ at time $1$ is at most $\frac{11}{36}$.
	    We obtain the expected profit $q_1w_1=\frac{2}{9}$ from $v_1$.
	    Since $v_2$ is the last vertex to arrive, we should assign $v_2$ if $u$ is available at time $2$. 
	    The probability that $u$ is available at time $2$ is $q_1 \Pr[C_t=1]=\frac{1}{4}$ because it is equal to the probability that $u$ accepts $v_1$ and $C_1$ takes $1$. 
	    Thus the probability that $v_2$ is assigned to $u$ at time $2$ is 
        $\Pr[(u \text{\ is available at time 2}) \wedge (v_2\text{\ arrives at time\ } 2)] = \frac{1}{4}$.
        Hence the expected profit at time $2$ is $\frac{6}{9} \cdot \frac{1}{4} \cdot \frac{1}{2}=\frac{1}{12}$.
        Therefore, the expected profit in total is $\frac{2}{9}  +  \frac{1}{12} = \frac{11}{36}$.
        On the other hand, the expected profit obtained when we discard $v_1$ is $\frac{1}{3}$, and thus we see that it is better to discard $v_1$. 
	    Since the probability that $v_3$ arrives is $\varepsilon$, 
        the offline optimal value is $\varepsilon \cdot \frac{4}{9\varepsilon}+ (1-\varepsilon)\cdot \frac{3}{9} = \frac{7-3\varepsilon}{9}$.
	    
	    Let us see a corresponding feasible solution $x$ to the offline LP.
	    We denote $x_{uv_t, t}=x_t$ for $t=1,2,3$ for simplicity.
	    Then as in the proof of Lemma~\ref{lem:offlineOPT}, $x_t$ corresponds to the probability that $v_t$ is assigned to $u$.
	    Since $v_1$ is not chosen in any offline optimal algorithm, $x_1$ is $0$.
        Moreover, $x_2$ is $\varepsilon \cdot 0 + (1-\varepsilon)\cdot 1 = 1-\varepsilon$ since we choose $v_2$ if $v_3$ does not arrive. 
	    Finally, $x_3$ is $\varepsilon\cdot 1 + (1-\varepsilon)\cdot 0 = \varepsilon$. 
	    The LP objective value for this solution $x$ is $\frac{2}{9}x_1+\frac{3}{9}x_2+\frac{4}{9\varepsilon}x_3 = \frac{7-3\varepsilon}{9}$.
		\end{example}

\subsection{Competitive Ratio}\label{sec:compratio}
We evaluate the performance of an online algorithm by a competitive ratio.
Let $\ALG(I)$ be the expected profit of an online algorithm $\ALG$ when the input sequence is $I$.
We say that an online algorithm is \emph{$\alpha$-competitive} if $\mathbb{E}_{I \sim \cI}[\ALG(I)] \geq \alpha \mathbb{E}_{I\sim \cI}[\OPT(I)]$ for any instance.

		\section{Proposed Algorithm}
		In this section, we present our algorithm,
		and then analyze the competitive ratio in subsequent subsections. 
		
		The overview of our proposed algorithm is described as follows.
		We first find an optimal solution $x^*$ to LP (\texttt{Off}).
		Then we use $x^*_{e,t}$ to determine a probability that $v$ comes at time $t$ and we assign $u$ to $v$ where $e=(u, v)$.
		That is, we choose a set $S$ of at most $b_v$ vertices in $U$ so that $\Pr [u\in S, v \text{\ arrives at time\ } t] = x^*_{e,t}$. 
		Then, for each $u\in S$, we compute the expected profit earned at and after $t$ by assigning $u$ to $v$,
		and, if it is larger than the one by not assigning, then the algorithm tries to assign $u$ to $v$.

        Our algorithm is summarized in Algorithm~\ref{alg}.
		We now explain how to implement each step in more detail.
		
		\paragraph{Finding a Set of Offline Vertices}
		We first explain how to design a probabilistic distribution to find a set of offline vertices when $v$ arrives at time $t$.
		We note that it is easy when $b_v=1$, as we can just choose a vertex $u$ in $U$ with probability $x^*_{uv, t}/p_{v, t}$.
		However, when $b_v\geq 2$, an independently random choice of offline vertices may violate the feasibility constraint.
		To avoid it, we construct a probability distribution over feasible vertex sets and choose a feasible vertex set according to the distribution.

        Let $\mathcal{S}_v = \{S \subseteq E_v \mid |S|\leq b_v\}$. 
        For $S \in \mathcal{S}_v$, its characteristic vector is a vector $\chi_{S}\in \{0,1\}^{E_v}$ such that $(\chi_{S})_e=1$ if and only if $e\in S$.
        Consider the convex hull $P_v$ of all characteristic vectors $\chi_S$ ($S\in \mathcal{S}_v$).
        Then the convex hull coincides with the polytope $\{ y \in [0,1]^{E_v} \mid \sum_{e\in E_v} y_{uv} \leq b_v \}$.
        In addition, every vertex in $P_v$ is an integral vector, which corresponds to a characteristic vector $\chi_S$ for some  $S\in \mathcal{S}_v$.
        
        For an optimal solution $x^*$ of LP~(\texttt{Off}), 
        define $y^{v,t}=(x^*_{e,t}/p_{v,t})_{e \in E_v}$.
        Then we can see that $y^{v,t} \in P_v$ since $x^*$ satisfies \eqref{eq:offlineLP-3}. 
        It follows from well-known Carath\'eodory theorem that $y^{v,t}$ can be decomposed as a convex combination of at most $|E_v|+1$ vertices in $P_v$. 
        Since every vertex in $P_v$ corresponds to a characteristic vector of a feasible set, there exist $S^v_k\in \mathcal{S}$ and $\lambda^{v,t}_k$~($k=1,\dots, |E_v|+1$) such that
\begin{align}
\textstyle
    y^{v,t} = \sum_{k=1}^{|E_v|+1} \lambda^{v,t}_{k} \chi_{S^v_k},
    \label{eq:decomp_y}
\end{align}
where $\sum_{k}\lambda^{v,t}_{k} = 1$ and $\lambda^{v,t}_{k} \geq 0$ for any $k\in [|E_v|+1]$.

        We regard $\lambda^{v,t}_{k}$~($k=1,\dots, |E_v|+1$) as a probability distribution over $\mathcal{S}_v$. 
        That is, we choose a set $S_k \in \mathcal{S}_v$ with probability $\lambda^{v,t}_k$. 
        Then the probability that $u$ is chosen when $v$ arrives at $t$ is $\sum_{k: u\in S^v_k} \lambda^{v,t}_{k}= y^{v,t}_{uv} = \frac{x^*_{uv,t}}{p_{v,t}}$.

    We note that $\lambda^{v,t}_k$'s can be obtained in polynomial time by the constructive proof of Carath\'eodory theorem. 
    Indeed, this can be done by the following procedure:
\begin{enumerate}
\item $y \leftarrow y^{v,t}$ and $k\leftarrow 1$;
\item while $y$ is not zero,
\begin{enumerate}
    \item let $S_k$ be a set of offline vertices with the $b_v$ largest positive value of $y_u$ (if only less than $b_v$ vertices with positive values exist, then take all);
    \item $\lambda^{v,t}_k = \min_{u \in S_k} y_u$, $y \leftarrow y - \lambda^{v,t}_k \chi_{S_k}$ and $k\leftarrow k+1$.
\end{enumerate}
\end{enumerate}

		\begin{remark}
 		In the above, we use only the fact that a subset of a feasible set is also feasible.
		Thus our algorithm would work even for more general constraints.
		We note, however, that it is required to solve linear programming problem over the convex hull $P_v$ with constraints~\eqref{eq:offlineLP-1}, \eqref{eq:offlineLP-2}, and~\eqref{eq:offlineLP-4}.
        This implies that, to run our algorithm in polynomial time, we need to describe $P_v$ efficiently.
		An example is a matroid constraint, that is, each online vertex $v$ has a matroid on the ground set $U$ and $v$ chooses an independent set of the matroid.
		We note that a cardinality constraint is a special case of a matroid constraint.
		\end{remark}
		
		\paragraph{Assigning Offline Vertices}

        We next describe how to decide whether we assign each online vertex $u\in S^v_k$ to $v$. 
		
		Let $R^d_{u,t}$ be the expected profit that $u$ earns at and after time $t$, when $u$ has a remaining budget $d$.
		By definition, $R^0_{u,t}=0$.
		Also, for our notational convenience, we assume that $R^d_{u, t}=0$ for all $d > T$.
        If we assigned an online vertex $v$ to $u$ with remaining budget $d>1$ at time $t$, then the expected profit that $u$ earns at and after $t$ would be 
		\begin{align*}
		Q^d_{e,t} := q_{e}\left(w_{e}+\sum_{\ell=1}^{T-t} \Pr[C_e = \ell] R^d_{u,t+\ell}\right) +(1-q_{e})R^{d-1}_{u,t+1}.
		\end{align*}
		Otherwise, the expected profit will be $R^d_{u, t+1}$. 

		In our algorithm, for each vertex $u\in S^v_k$, we compare $Q^d_{uv,t}$ and $R^d_{u, t+1}$. 
		If the former is larger, then this means that assigning $u$ to $v$ makes larger profit than not assigning, and thus we assign $u$ to $v$.
		Otherwise, i.e., if assigning $u$ to $v$ does not make enough profit, then we do not assign $u$ to $v$.
		Thus the expected profit for a vertex $u\in S^v_k$ is $\max \left\{ Q^d_{uv,t}, R^d_{u,t+1} \right\}$.
		
		Since the probability that $v$ arrives at time $t$ and $u$ is chosen by $v$ is $x^*_{e,t}$ for each $e=(u, v)\in E_v$ and $t\in [T]$,  $R^d_{u,t}$ can be represented recursively by
		\begin{align}
		R^d_{u,t}=\sum_{e\in E_u} x^*_{e,t} \max \left\{ Q^d_{e,t}, R^d_{u,t+1} \right\} + \left(1-\sum_{e\in E_u} x^*_{e,t}\right) R^d_{u,t+1}.
		    \label{eq:R^d_t}
		\end{align}

		Our algorithm needs to compute $R^d_{u,t}$ for each $d$, $u\in U$ and $t\in [T]$.
        This can be done efficiently in advance by dynamic programming with the above recursive equation~\eqref{eq:R^d_t}.

        \begin{algorithm}[tb]
		    \caption{Proposed online algorithm}\label{alg}
		    \begin{algorithmic}[1]
		        \STATE Solve LP~(\texttt{Off}) to obtain an optimal solution $x^*$
		        \FOR{$t=1,\ldots, T$}
		            \STATE Let $v$ be a vertex that arrives at time $t$~(if none, skip)
    		        \STATE Choose $S^v_k$ with probability $\lambda^{v,t}_k$ by ~\eqref{eq:decomp_y}
    		        \FOR{$u\in S^v_k$}
    		            \STATE Let $d$ be the remaining budget of $u$
    		            \IF {$Q^d_{uv, t} \geq R^d_{u,t+1}$ and $u$ is available} \label{alg:line}
    		            \STATE Assign $u$ to $v$
    		            \ELSE
    		            \STATE Do nothing for $u$
    		            \ENDIF
    		        \ENDFOR
    		    \ENDFOR
		    \end{algorithmic}
		\end{algorithm}
		
		\subsection{Modifying Instance}\label{sec:modify}

		We first observe that the total expected profit of the algorithm is equal to $\sum_{u\in U} R^{\Delta_u}_{u,1}$.
		Since LP (\texttt{Off}) gives an upper bound on $\mathbb{E}_{I\sim \cI}[\OPT(I)]$ by Lemma~\ref{lem:offlineOPT}, we aim to determine the ratio between $\sum_{u\in U} R^{\Delta_u}_{u,1}$ and the LP optimal value $\sum_{u \in U} \sum_{e\in E_u} \sum_{t} w_e q_e x^*_{e,t}$.
		To this end, we fix a vertex $u$ in $U$, and we evaluate the ratio $\alpha_u$ between $R^{\Delta_u}_{u,1}$ and $\sum_{e\in E_u} \sum_{t} w_e q_e x^*_{e,t}$. 
		Then the competitive ratio of Algorithm~\ref{alg} is at least $\min_u \alpha_u$.
		
		To obtain a lower bound of $R^{\Delta_u}_{u,1}$, we modify a given instance to a simpler one with LP optimal solution $x^*$. 
		We will show that the expected profit for the modified instance gives a lower bound on that for the original instance.

		The new instance is defined as follows.
		We define $G'=(U', V'; E')$ where $U'=\{u\}$, $V'=\{v_t\mid t\in [T]\}$, and $E'=\{(u, v_t)\mid t\in [T]\}$.
		In this instance, we have only one offline vertex $u$, and, at each $t \in [T]$, only one vertex $v_t$ may arrive with edge $(u, v_t)$.
		We set the parameters\footnote{For ease of notation, we use simpler subscripts, e.g., $p'_t$ instead of $p_{uv_t,t}$, as we have only one online vertex at time $t$.} for $v_t$ in the new instance as follows:
		\begin{itemize}
		    \item the probability of arrival: $p'_t = \sum_{e\in E_u} x^*_{e,t}$;
		    \item the probability of acceptance: $q'_t = \sum_{e\in E_u} x^*_{e,t}q_{e}/p'_t$ if $p'_t>0$, and $q'_t=0$ otherwise;
		    \item the profit: $w'_t = \sum_{e\in E_u} x^*_{e,t} q_{e}w_{e}/(p'_tq'_t)$ if $p'_tq'_t>0$ and $0$ otherwise;
		  \item the distribution of occupation time (denoted by $C_t$):
		  \[
		  \Pr[C_t=\ell] = \frac{\sum_{e\in E_u} x^*_{e,t}q_{e}\Pr[C_{e}=\ell]}{p'_tq'_t} \quad (\ell \in [T])
		  \]
		  if $p'_t,q'_t > 0$, and otherwise, $\Pr[C_t=T]=1$;
		  \item the task capacity: $b'_t = 1$.
		\end{itemize}
		The budget $\Delta_u$ of $u$ is set the same value as in the original instance.
		Since $x^*$ is a feasible solution to LP {\rm (\texttt{Off})}, it holds that $p'_t \leq 1$ and $q'_t \leq 1$ by \eqref{eq:offlineLP-4}, and moreover, 
		\[
    		\sum_{\ell=1}^T\Pr[C_t = \ell] = \frac{\sum_{e\in E_u} x^*_{e,t}q_{e} \sum_{\ell}\Pr[C_{e}=\ell]}{p'_tq'_t} = 1.
		\]
		Note also that $\sum_{t\in [T]}w'_t q'_t p'_t = \sum_{e\in E_u} \sum_{t} w_e q_e x^*_{e,t}$.

		Let us execute Algorithm~\ref{alg} for the modified instance, where we use $p'_t$ instead of LP optimal solution $x^*$. 
		Let $\tilde{R}^d_t$ be the expected profit that $u$ earns on and after $t$ when $u$ has a remaining budget $d$ at time $t$.
		Similarly to~\eqref{eq:R^d_t} for $R^d_{u,t}$, it holds that 
		\begin{align}\label{eq:R}
			\tilde{R}^d_{t} = p'_t \max\left\{ Q^d_t, \tilde{R}^d_{t+1} \right\} +(1-p'_t) \tilde{R}^d_{t+1},
		\end{align}
		where 
		\begin{align}\label{eq:Q}
		Q^d_t = q'_t\left(w'_t + \sum_{\ell=1}^{T-t} \Pr[C_t = \ell] \tilde{R}^d_{t+\ell}\right) +(1-q'_t)\tilde{R}^{d-1}_{t+1}.
		\end{align}
		
		We show that the modification does not increase the profit.

		\begin{lemma}\label{lem:replace}
		$R^d_{u,t} \geq \tilde{R}^d_{t}$ for all $d\in [\Delta_u]\cup \{0\}$ and $t\in [T]$.
		\end{lemma}
		\begin{proof}
		We show the lemma by the backward induction on $t$.
		For the base case, we observe that $\tilde{R}^d_{T}=p'_Tq'_tw'_t = \sum_{e} x^*_{e,t}q_e w_e =  R^d_{u,T}$. 
		Suppose that $R^d_{u,\ell} \geq \tilde{R}^d_{\ell}$ for all $\ell=t+1,\ldots, T$. 
		By definition of $R^d_{u,t}$,
		\begin{align*}
		    R^d_{u,t} &= \sum_{e\in E_u} x^*_{e,t} \max \left\{ Q^d_{e,t}, R^d_{u,t+1} \right\} + \left(1-\sum_{e\in E_u} x^*_{e,t}\right) R^d_{u,t+1} \\
		    &\geq \max \left\{ \sum_{e\in E_u} x^*_{e,t}Q^d_{e,t}, p'_t R^d_{u, t+1} \right\} + (1-p'_t) R^d_{u,t+1}\\
		    &\geq \max \left\{ \sum_{e\in E_u} x^*_{e,t}Q^d_{e,t}, p'_t \tilde{R}^d_{t+1} \right\} + (1-p'_t) \tilde{R}^d_{t+1},
		 \end{align*}
		 where the last inequality holds by the induction hypothesis.
		 Moreover, by the induction hypothesis, we have
		 \begin{align*}
		     \sum_{e\in E_u} x^*_{e,t} Q^d_{e,t}
		     &= p'_t q'_t w'_t + p'_tq'_t \sum_{\ell=1}^{T} \Pr[C_t = \ell] R^d_{u,t+\ell} + p'_t(1-q'_t)R^{d-1}_{u,t+1}\\
 		     &\geq p'_t \left( q'_t w'_t + q'_t \sum_{\ell=1}^{T} \Pr[C_t = \ell] \tilde{R}^d_{t+\ell} + (1-q'_t)\tilde{R}^{d-1}_{t+1}\right)\\
		     & = p'_t Q^d_t.
		 \end{align*}
		 Therefore, it holds that
		 \begin{align*}
		     R^d_{u,t} 
		     &\geq p'_t \max\left\{ Q^d_t, \tilde{R}^d_{t+1} \right\} +(1-p'_t) \tilde{R}^d_{t+1} 
		     =\tilde{R}^d_{t}.
		 \end{align*}
		Thus the lemma holds. 
		\end{proof}
		
		By this lemma, it suffices to lower-bound $\tilde{R}^{\Delta_u}_1$ to obtain a lower bound on $R^{\Delta_u}_{u,1}$ for the original instance. 
		
		\subsection{Analysis for the Unlimited Rejection Case}\label{sec:unlimited}
		
		In this section, we assume that the number of allowed rejections is unlimited, that is, $\Delta_u =+\infty$ for $u\in U$.
		This means that we can ignore the constraint~\eqref{eq:offlineLP-2} in LP (\texttt{Off}).
		The main result of this section is to prove that Algorithm~\ref{alg} is $1/2$-competitive for this case.
		For that purpose, we first modify the instance as in the previous section, and we will bound $\tilde{R}^{\Delta_u}_{1}$ from below.
		
		We observe that, when executing the algorithm for the modified instance, the remaining budget $d$ of the vertex $u$ is infinite in the whole process. 
		For simplicity, we denote $R_t=\tilde{R}^{\infty}_{t}$ and $Q_t=Q^{\infty}_t$.
        Then, by~\eqref{eq:Q}, it holds that
		\begin{align*}
		Q_t & = q'_t\left( w'_t + \sum_{\ell =1}^{T-t} \Pr[C_t=\ell] R_{t+\ell}\right)+(1-q'_t)R_{t+1}.
		\end{align*}
		Moreover, it follows from \eqref{eq:R} that
		\begin{align}\label{eq:Rt}
		    R_t 
		     & = \max \left\{ p'_t Q_t + (1-p'_t) R_{t+1}, R_{t+1} \right\}.
		\end{align}
		We note that $R_T = p'_Tq'_Tw'_T$.
		We can make the first term in~\eqref{eq:Rt} simpler as 
	 	$
        B_t w'_t + A_tR_{t+1} + B_t \sum_{\ell =2}^{T-t} \Pr[C_t=\ell] R_{t+\ell},
	 	$
 		where $B_t=p'_tq'_t$ and $A_t = p'_tq'_t \Pr[C_t=1]+p'_t(1-q'_t)+(1-p'_t)$ for each $t \in [T]$. 
 		We note that $A_t=p'_tq'_t \Pr[C_t=1] + 1- p'_tq'_t = 1-B_t\Pr[C_t \geq 2]$.

		A lower bound on $R_1$ is obtained by minimizing $R_1$ subject to the condition~\eqref{eq:Rt}.
		The minimization problem can be formulated as a linear programming problem as below.
		Note that we do not need to solve the LP, but use it for analysis.\footnote{Such LP is called a factor-revealing LP.}
		\begin{align}
		    \begin{array}{rl}
		        \min & R_1 \\
		        \text{s.t.} & R_{t} \geq B_t  w'_t + A_t R_{t+1} +  B_t\sum_{\ell =2}^{T-t} \Pr[C_t=\ell] R_{t+\ell}  \ \ (t \in[T-1])\\
		        & R_t \geq R_{t+1} \quad  (t \in [T-1]) \\
		        & R_T \geq B_T w'_T \\
		        & R_t \geq 0 \quad ( t\in [T]),
		    \end{array}
		    \label{eq:LP-alg-primal-1}
		\end{align}
		Then the optimal value of the above LP gives a lower bound of $R_1$.
		The dual of LP \eqref{eq:LP-alg-primal-1} is given by
		\begin{align}
		    \begin{array}{rl}
		        \max & \sum_{t=1}^T B_t w'_t \alpha_t\\
		        \text{s.t.}  & \alpha_1+\beta_1 \leq 1\\
		        & \alpha_2+\beta_2\leq A_1\alpha_1+\beta_1 \\
		        & \alpha_t+\beta_t\leq \sum_{\ell =1}^{t-2} \alpha_\ell B_\ell \Pr[C_\ell = t-\ell]+A_{t-1}\alpha_{t-1} + \beta_{t-1} \ \  (3 \leq t \leq T-1)\\
		        & \alpha_T \leq \sum_{\ell =1}^{T-2} \alpha_\ell B_\ell \Pr[C_\ell = T-\ell]+A_{T-1}\alpha_{T-1}+\beta_{T-1} \\
		        &\alpha_t \geq 0 \quad (t\in [T])\\
		        &\beta_t \geq 0 \quad (t\in [T-1]).
		    \end{array}
		    \label{eq:LP-alg-dual-1}
		\end{align}
		To show a lower bound on LP~\eqref{eq:LP-alg-primal-1},
		we construct a feasible solution to the dual LP \eqref{eq:LP-alg-dual-1}. 
		Let $\gamma \leq 1/2$.
		We set $\alpha_t=\gamma$ for $t \in[T]$, $\beta_1=1-\gamma$, $\beta_2 = \beta_1 -\gamma +A_1 \gamma$, and $\beta_t = \beta_{t-1}-\gamma+A_{t-1}\gamma+\sum_{\ell =1}^{t-2} B_\ell \Pr[C_\ell = t-\ell] \gamma$ ($3 \leq t\leq T-1$).

        \begin{lemma}\label{lem:feasible-unlimited}
        For $0\leq \gamma\leq 1/2$, $\alpha_t$ and $\beta_t$ defined as above are feasible to \eqref{eq:LP-alg-dual-1}.
        \end{lemma}
        \begin{proof}
        By the construction, we can see that they satisfy the first, second and third constraints in \eqref{eq:LP-alg-dual-1}.
        In fact, these constraints are satisfied with equality as $\alpha_t = \gamma$ for any $t$.
        It remains to show that $\beta_t$ is non-negative, and that they satisfy the fourth constraint.
        
        \begin{claim}
        For $t\geq 2$, we have
        \begin{align}\label{eq:betaLB}
         \beta_t = 1-\gamma - \gamma \sum_{t'<t} B_{t'}\Pr[C_{t'}\geq t-t'+1].
        \end{align}
        \end{claim}
        \begin{proof}[Proof of Claim]
        By a simple calculation, we have
		\begin{align*}
		    \beta_2 &=1-\gamma -\gamma   + (1-B_1\Pr[C_1\geq 2])\gamma \\
		    &= 1-\gamma - \gamma \sum_{t'<2} B_{t'} \Pr[C_{t'}\geq 2-t'+1].
		\end{align*}        
		Moreover, by induction, it holds that, for any $t\geq 3$,  
		\begin{align*}
		    \beta_t &=\beta_{t-1} +\gamma \sum_{\ell=1}^{t-2}B_\ell \Pr[C_\ell = t-\ell] - (1-A_{t-1})\gamma\\
		    &= \left(1-\gamma - \gamma \sum_{t'<t-1} B_{t'}\Pr[C_{t'}\geq t-t']\right) +\gamma \sum_{t'<t-1}B_{t'} \Pr[C_{t'} = t-t'] - (B_{t-1} \Pr[C_{t-1} \geq 2])\gamma\\
		    &= 1-\gamma - \gamma \sum_{t'<t} B_{t'}\Pr[C_{t'}\geq t-t'+1].
		\end{align*}		
		Thus the claim holds.
        \end{proof}

		By~\eqref{eq:betaLB}, 
		we see that, for $t=2,\ldots, T-1$,
		\begin{align*}
		    \beta_t  &=1-\gamma - \gamma \sum_{t'<t} B_{t'}\Pr[C_{t'}\geq t-t'+1] \\
		    &=1-\gamma - \gamma\left( \sum_{t'<t} \sum_{e\in E_u}x^*_{e,t'} q_{e}\Pr[C_{e}\geq t-t'+1]\right)\\
		    &\geq 1-2\gamma \quad (\because \eqref{eq:offlineLP-1}).
		\end{align*}
		Since $\gamma \leq 1/2$, it turns out that $\beta_t$ is nonnegative for any $t$.

		Finally, let us check that the fourth constraint is satisfied.
		The RHS is equal to $ - \gamma (1-A_{T-1} - \sum_{\ell=1}^{T-2} B_\ell \Pr[C_\ell = T-\ell])$.
		Since $A_{T-1}=1-B_{T-1}\Pr [C_{T-1}\geq 2]$, the RHS is at least $ - \gamma B_{T-1}\Pr [C_{T-1}\geq 2]$.
		Since $\beta_{T-1}=1-\gamma - \gamma \sum_{t'<T-1} B_{t'}\Pr[C_{t'}\geq T-t'+1]$ by~\eqref{eq:betaLB}, the fourth constraint  is satisfied if  
		\begin{align}
		    1 &\geq \left(1+\sum_{t'<T} B_{t'} \Pr[C_{t'}\geq T-t'+1]\right)\gamma\notag \\ 
		    &= \left(1+\sum_{t'<T} \sum_{e\in E_u}x^*_{e,t'} q_{e}\Pr[C_{e}\geq T-t'+1]\right)\gamma. \label{eq:unlimited-0}
		\end{align}
		Since $0\leq \sum_{t'<T} \sum_{e\in E_u}x^*_{e,t'} q_{e}\Pr[C_{e}\geq T-t'+1] \leq 1$ by \eqref{eq:offlineLP-1}, the RHS is at most $2\gamma$. 
		Since $\gamma\leq 1/2$, \eqref{eq:unlimited-0} holds, and hence the fourth constraint is satisfied.
        \end{proof}

		Therefore, $\alpha_t$ and $\beta_t$ defined as above are feasible to \eqref{eq:LP-alg-dual-1}.
		The objective value for this solution is $\sum_{t} B_t w'_t \gamma = \gamma \sum_{t} p'_t q'_t w'_t$.
		It follows from the LP duality theorem that the optimal value of LP~\eqref{eq:LP-alg-primal-1} is at least $\gamma \sum_{t} p'_t q'_t w'_t$.
		This is maximized when $\gamma = 1/2$, and in this case, $R_1$ is lower-bounded by $\frac{1}{2} \sum_{t} p'_t q'_t w'_t$.
		Since $\sum_{t} p'_t q'_t w'_t = \sum_{t} \sum_{e\in E_u} w_{e}q_ex^*_{e,t}$, we obtain $R_1 \geq \frac{1}{2} \sum_{t} \sum_{e\in E_u} w_{e}q_ex^*_{e,t}$.
		
		Summarizing, we have $\sum_{u\in U} R^{\Delta_u}_1 \geq \frac{1}{2} \mathbb{E}[\OPT(I)]$ by Lemma~\ref{lem:offlineOPT}, which implies the following theorem.

		\begin{theorem}\label{thm:main-unlimited}
		Algorithm~\ref{alg} is $1/2$-competitive for the problem with unlimited rejections.
		\end{theorem}
		
		\subsection{Analysis for the Limited Rejection Case}\label{sec:limited}
		
		In this subsection, we prove that Algorithm~\ref{alg} is $\Delta/(3\Delta-1)$-competitive for the general case with $\Delta = \max_{u \in U}\Delta_u$. 
		As in the analysis for the unlimited rejection case~(Section~\ref{sec:unlimited}),
		we present a lower bound on $\tilde{R}^d_t$.
		
		We first show the following lemma. 
		\begin{lemma}\label{lem:general}
		$\tilde{R}^{d-1}_t \geq \frac{d-1}{d}\tilde{R}^{d}_t$ for each $d\in [\Delta_u]$ and $t \in [T]$.
		\end{lemma}
		\begin{proof}
		We prove the lemma by induction on $t$. 
		The base case is when $t=T$. In this case,
	    $\tilde{R}^d_T=\tilde{R}^{d'}_T$ holds for any $d,d'\in [\Delta_u]$, and thus the statement holds.
		
		Let us consider the case when $t < T$. When $d=1$, the statement is trivial.
		Hence we prove that the statement holds for $t$ and $d> 1$.

		By \eqref{eq:R}, we have $\tilde{R}^d_{t}=\max(p'_{t} Q^d_{t}+(1-p'_{t})\tilde{R}^d_{t+1}, \tilde{R}^d_{t+1})$.
		\begin{align*}
		    \tilde{R}^d_{t} &= \max \{p'_tQ^d_t+(1-p'_t)\tilde{R}^d_{t+1}, \tilde{R}^d_{t+1}\}
		    \\
		    &\geq \max \left\{\frac{d}{d+1}(p'_tQ^d_t+(1-p'_t)\tilde{R}^d_{t+1})+\frac{1}{d+1}\tilde{R}^d_{t+1}, \tilde{R}^d_{t+1}\right\} 
		\end{align*}
		By the induction hypothesis, $\tilde{R}_{t+1}^{d-1} \ge \frac{d-1}{d} \tilde{R}_{t+1}^{d}$ holds. This shows that
		$Q^d_t = q'_t(w_t+\sum_{\ell \geq 1} \Pr[C_t = \ell] \tilde{R}^d_{t+\ell}) + (1-q'_t)\tilde{R}^{d-1}_{t+1} \geq q'_t(w'_t+\sum_{\ell \geq 1} \Pr[C_t = \ell] \tilde{R}^d_{t+\ell}) + (1-q'_t) \frac{d-1}{d}\tilde{R}^{d}_{t+1}$.
		Also we can decompose $\tilde{R}^d_{t+1}
		    = p'_tq'_t \sum_{\ell \geq 1} \Pr[C_t = \ell]\tilde{R}^d_{t+1}
		    + p'_t(1-q'_t) \tilde{R}^d_{t+1}
		    + (1-p'_t) \tilde{R}^d_{t+1}$.
		Hence, it holds that
		\begin{align*}
		    &\frac{d}{d+1}(p'_tQ^d_t+(1-p'_t)\tilde{R}^d_{t+1})+\frac{1}{d+1}\tilde{R}^d_{t+1} \\
		    &\geq \frac{d}{d+1}p'_tq'_tw'_t + p'_tq'_t\sum_{\ell \geq 1} \Pr[C_t = \ell]\left(\frac{d}{d+1} \tilde{R}^d_{t+\ell} +  \frac{1}{d+1}\tilde{R}^d_{t+1}\right) \\
		& \quad
		+ p'_t(1-q'_t) \left(\frac{d-1}{d+1}\tilde{R}^{d}_{t+1} + \frac{1}{d+1} \tilde{R}^d_{t+1} \right)
		+ (1-p'_t)\tilde{R}^d_{t+1}\\
		&\geq \frac{d}{d+1}p'_tq'_tw'_t 
		+ p'_tq'_t\sum_{\ell \geq 1} \Pr[C_t = \ell] \tilde{R}^d_{t+\ell} \\
		&\quad + p'_t(1-q'_t) \frac{d}{d+1}\tilde{R}^d_{t+1}+ (1-p'_t)\tilde{R}^d_{t+1}, 
		\\
		&\geq \frac{d}{d+1} \left( p'_t Q^{d+1}_t+(1-p'_t)\tilde{R}^{d+1}_{t+1}\right)
		\end{align*}
		The second inequality holds because $\tilde{R}^d_{t+1} \geq \tilde{R}^d_{t+\ell}$ ($\ell \geq 1$), and the last inequality follows from the induction hypothesis. 
		Therefore, we have 
		\begin{align*}
		    \tilde{R}^d_t &\geq  \max \left\{
		    \frac{d}{d+1}(p'_tQ^d_t+(1-p'_t)\tilde{R}^d_{t+1})+\frac{1}{d+1}\tilde{R}^d_{t+1}, 
		    \tilde{R}^d_{t+1} \right\}\\
		    &\geq \frac{d}{d+1} \max \left( p'_t Q^{d+1}_t+(1-p'_t)\tilde{R}^{d+1}_{t+1}, \tilde{R}^d_{t+1} \right)\\
		    &= \frac{d}{d+1} \tilde{R}^{d+1}_{t}.
		\end{align*}
		\end{proof}

		By~\eqref{eq:R} with the above lemma, we have the following relationships:
		\begin{align}
		    \tilde{R}^d_{t} \geq  \max\left\{ p_t\hat{Q}^d_t+(1-p_t)\tilde{R}^d_{t+1}, \tilde{R}^d_{t+1}\right\},
		    \label{eq:R_replaced}
		\end{align}
		where $\hat{Q}^d_t = q_t(w_t+\sum_{\ell \geq 1} \Pr[C_t = \ell] \tilde{R}^d_{t+\ell}) + \frac{d-1}{d}(1-q_t)\tilde{R}^{d}_{t+1}$.
		
		Recall that it suffices to give a lower bound on $\tilde{R}^{\Delta_u}_1$ 
		to obtain the competitive ratio of Algorithm~\ref{alg}.
		By \eqref{eq:R_replaced}, we can construct a linear program to bound $\tilde{R}^{\Delta_u}_1$ in a similar way to the previous subsection.

		\begin{theorem}\label{thm:main-general}
		Algorithm~\ref{alg} is a $\Delta/(3\Delta-1)$-competitive algorithm, where $\Delta = \max_{u \in U} \Delta_u$.
		\end{theorem}
		\begin{proof}
		By the discussion in Section~\ref{sec:limited}, we show a lower bound on $\tilde{R}^{\Delta_u}_1$ given in \eqref{eq:R_replaced}.
		By \eqref{eq:R_replaced}, the optimal value of the following LP is a lower bound on $\tilde{R}_1^{\Delta_u}$.
		\begin{align}
		    \begin{array}{rl}
		        \min  & R_1  \\
		        \text{s.t.}  & R_{t} \geq B_t  w'_t + A_t R_{t+1} +  B_t\sum_{\ell =2}^T \Pr[C_t=\ell] R_{t+\ell} \ \ (t \in [T-1])\\
		        & R_t \geq R_{t+1} \quad  ( t \in [T-1]) \\
		        & R_T \geq B_T w'_T  \\
		        & R_t \geq 0 \quad ( t \in [T]),
		    \end{array}
		    \label{eq:LP-alg-primal-2}
		\end{align}
		where $B_t=p'_tq'_t$ and $A_t = p'_tq'_t \Pr[C_t=1]+p'_t(1-q'_t)\frac{\Delta_u-1}{\Delta_u}+(1-p'_t)$.
		The dual of LP \eqref{eq:LP-alg-primal-2} is
		\begin{align}
		    \begin{array}{rl}
		        \max  & \sum_{t} B_t w'_t \alpha_t \\ 
		        \text{s.t.} & \alpha_1+\beta_1 \leq 1\\
		        & \alpha_2+\beta_2 \leq A_1\alpha_1+\beta_1 \\
		        & \alpha_t+\beta_t \leq \sum_{\ell =1}^{t-2} B_\ell \Pr[C_\ell = t-\ell] \alpha_\ell +A_{t-1}\alpha_{t-1}+\beta_{t-1} \  (3 \leq t \leq T-1) \\
		        & \alpha_T \leq \sum_{\ell =1}^{T-2} B_\ell \Pr[C_\ell = T-\ell] \alpha_\ell +A_{T-1}\alpha_{T-1}+\beta_{T-1}\\
		        &\alpha_t \geq 0 \quad (t \in [T])\\
		        &\beta_t \geq 0 \quad (t \in [T-1]).
		    \end{array}
		    \label{eq:LP-alg-dual-2}
		\end{align}
		In a similar way to the previous subsection, we present a feasible solution to LP~\eqref{eq:LP-alg-dual-2}.
		We set $\alpha_t=\gamma$ ($t\leq T$), $\beta_1=1-\gamma$, $\beta_2 = 1-\gamma - (1-A_1)\gamma$. 
		For $t \geq 3$, we set
		\begin{align}
		    \beta_t
		        &= 1-\gamma - \gamma \cdot \frac{1}{\Delta_u} \sum_{\ell <t} p'_\ell (1-q'_\ell)- \gamma \sum_{\ell <t} p'_\ell q'_\ell \Pr[C_\ell \geq t-\ell+1]. \label{eq:beta}
		    \end{align}
		It is not difficult to see that the above solution satisfies the constraints in \eqref{eq:LP-alg-dual-2} except the nonnegativity ones. 
Let us find $\gamma$ that makes the solution nonnegative.
By definitions of $\alpha_t$ ($t\in [T]$) and $\beta_1$, $\gamma$ has to satisfy $0 \leq \gamma \leq 1$.

For $\beta_2$, 
it holds that $1-A_1 \leq 1$ since $A_1 \geq 0$, 
and hence $\beta_2 \geq 1-2\gamma$. 

For $t \geq 3$, we have $\beta_t = 1-\gamma - \gamma \theta$, where
\begin{align*}
\theta &= \frac{1}{\Delta_u}\sum_{\ell <t} p'_\ell (1-q'_\ell) + \sum_{\ell <t} p'_\ell q'_\ell \Pr[C_\ell \geq t-\ell+1]\\
& = \frac{1}{\Delta_u}\sum_{\ell <t} \sum_{e \in E_u}  x^*_{e,\ell}\left\{(1-q_e) +q_e  \Pr[C_\ell \geq t-\ell+1]\right\} + \frac{\Delta_u-1}{\Delta_u}
\sum_{\ell <t} \sum_{e \in E_u} x^*_{e,\ell}q_e  \Pr[C_\ell \geq t-\ell+1]
\end{align*}	
Note that $\theta \leq 1 + \frac{\Delta_u-1}{\Delta_u}$ by the constraints of LP {\rm (\texttt{Off})}.
Therefore, 
the dual solution is feasible to \eqref{eq:LP-alg-dual-2} 
when $\gamma \ge 1/(2+(\Delta_u-1)/\Delta_u) = \Delta_u/(3\Delta_u-1)$, that achieves
the objective value $\frac{\Delta_u}{3\Delta_u-1}\cdot \sum_{t} p'_t q'_t w'_t$.
This implies that $\tilde{R}^{\Delta_u}_1 \geq \frac{\Delta_u}{3\Delta_u-1}\cdot \sum_{t} p'_t q'_t w'_t$.
		
Summarizing the above discussion, we have $\mathbb{E}[\ALG(I)] = \sum_{u\in U} R^{\Delta_u}_1  \geq \frac{\Delta}{3\Delta-1}\cdot \sum_{t} \sum_{e} w_e q_e x^*_{e,t} \geq\frac{\Delta}{3\Delta-1}  \mathbb{E}[\OPT(I)]$ by Lemma~\ref{lem:offlineOPT}.
\end{proof}

\subsection{Analysis for the KIID Model with Non-reusable Agents}

In this subsection, we also show that our algorithm can be $\frac{1}{2}\left(1-\frac{1}{e^2}\right)$-competitive in the KIID model with non-reusable agents. 
We note that our algorithm has a better competitive ratio than the $\frac{1}{e}$-competitive algorithm by \citet{Nanda2020}.

\begin{theorem}\label{thm:KIID}
There is a $(\frac{1}{2-1/\Delta}\left(1-\frac{1}{e^{2-1/\Delta}}\right))$-competitive algorithm for the problem in the KIID model with non-reusable resources, where $\Delta = \max_{u \in U} \Delta_u$.
\end{theorem}
\begin{proof}
First, we assume that $T\geq 2$ because when $T=1$, we can attain the offline optimal value in expectation by assigning an offline vertex with the highest expected profit.
Then in Algorithm~\ref{alg}, we choose an optimal solution $x^*$ to LP (\texttt{Off}) as the one satisfying $x^*_{e,1} =\cdots = x^*_{e,T}$ for any $e\in E$.
We can do this because a solution defined by $x'_{e,t}=\frac{\sum_\ell x^*_{e,\ell}}{T}$ is also an optimal solution to LP (\texttt{Off}).
Here, since $\Pr[C_e = T]=1$ for all $e\in E$ and $p_{v,1}=\cdots=p_{v,T}=p_v$ for each $v\in V$, LP (\texttt{Off}) is rewritten as
\begin{align}
			\text{max} \ & \sum_{e} w_e q_e \sum_{t\in [T]}x_{e,t}  \\ 
			\text{s.t.} \ & \sum_{e\in E_u} \sum_{t\in [T]}x_{e,t} q_e \leq 1 \quad (\forall u\in U)\\
			&  \sum_{e\in E_u} \sum_{t\in [T]} x_{e,t} \leq \Delta_u \quad (\forall u\in U) \label{eq:offlineLP-kiid}\\
			& \sum_{e\in E_v} x_{e,t} \leq p_{v}b_v \quad (\forall v\in V, t\in [T])  \\
			& 0\leq x_{e,t}\leq p_{v} \quad (\forall e\in E). 
\end{align}

Let us focus on the modified instance described in Section~\ref{sec:modify} for an offline vertex $u\in U$.
By the definition, $p'_t = \sum_{e\in E_u} x^*_{e,t}$ does not depend on $t$, and similarly for $q'_t$, $w'_t$.
For the occupation time $C_t$, we have $\Pr[C_t = t]=0$ for all $t<T$ and $\Pr[C_t = T] = 1$.
Thus, we write $p'_t = p'$, $q'_t=q'$ and $w'_t=w'$.

We can write the LP \eqref{eq:LP-alg-dual-2} as
\begin{align}
		    \begin{array}{rl}
		        \max \ & p'q' w' \sum_{t} \alpha_t \\ 
		        \text{s.t.}\ & \alpha_1+\beta_1 \leq 1,\\
		        & \alpha_t + \beta_t \leq A\alpha_{t-1} +\beta_{t-1} \ (2 \leq t \leq T-1), \\
		        & \alpha_T \leq A\alpha_{T-1}+\beta_{T-1}\\
		        &\alpha_t \geq 0 \quad (t \in [T]) \\
		        &\beta_t \geq 0 \quad (t \leq T-1),
		    \end{array}
		    \label{eq:LP-alg-dual-kiid}
		\end{align}
		where $A = p'(1-q')\frac{\Delta_u-1}{\Delta_u}+(1-p') = 1-p'/\Delta_u-p'q'(\Delta_u-1)/\Delta_u \geq 0$.

Similarly to Theorems~\ref{thm:main-unlimited} and \ref{thm:main-general}, we present a feasible solution to the above LP to show a lower bound on $R^{\Delta_u}_1$.
Let $\alpha_1 = 1$ and $\beta_1=0$.
For $t\geq 2$, let $\alpha_t = A\alpha_{t-1}$, and $\beta_t=0$.
It is clear that $\alpha_t = A^{t-1} \geq 0$ and $\beta_t \geq 0$ for all $t$.
We see that $\alpha_t+\beta_t = A\alpha_{t-1}+\beta_{t-1}$ for all $2 \leq t\leq T-1$, and $\alpha_T = A\alpha_{T-1}+\beta_{T-1}$.
Thus, this solution is feasible to LP~\eqref{eq:LP-alg-dual-kiid}.

We can prove the result by evaluating a ratio between $\sum_{t} p'q'w'\alpha_t$ and $\sum_{t} p'q'w' = p'q'w'T$.
For each $t\geq 1$, we have $\sum_t \alpha_t = \sum_{t} A^{t-1}$.
We observe that $p'T = \sum_{t}\sum_{e\in E_u} x^*_{e,t}\leq \Delta_u$ and $p'q'T =  \sum_t \sum_{e\in E_u} x^*_{e,t}q_e\leq 1$.
Thus, it holds that 
\[
A = 1-\frac{p'+(\Delta_u-1)p'q'}{\Delta_u} \geq 1-\frac{2-1/\Delta_u}{T}.
\]
We have $1-\frac{2-1/\Delta_u}{T} > 0$ because $T\geq 2$, and hence
\begin{align*}
    \frac{\sum_t \alpha_t}{T}
    &=\frac{1}{T}\sum_t A^{t-1}\\
    &\geq \frac{1}{T}\sum_t \left(1-\frac{2-1/\Delta_u}{T}\right)^{t-1} \\
    &=\frac{1}{2-1/\Delta_u}\left(1-\left(1-\frac{2-1/\Delta_u}{T}\right)^{T}\right)
    \geq \frac{1}{2-1/\Delta_u}\left(1-\frac{1}{e^{2-1/\Delta_u}}\right).
\end{align*}
Therefore, by the same discussion as Theorem~\ref{thm:main-general}, 
we see that
\begin{align*}
R^{\Delta_u}_1 &\geq \sum_{t} p'q'w'\alpha_t \\
&\geq \frac{1}{2-1/\Delta_u}\left(1-\frac{1}{e^{2-1/\Delta_u}}\right) \sum_{t\in [T]}\sum_{e\in E_u} w_e q_e x^*_{e,t}.
\end{align*}
This completes the proof of Theorem~\ref{thm:KIID}.
\end{proof}

		\section{Hardness}\label{sec:hardness}
In this section, we prove that our problem has competitive ratio at most $1/2$. 
In fact, this is easily shown with an example given by \cite{krengel1977semiamarts}.

\begin{theorem}\label{thm:hardness}
No online algorithm for the problem is $(\frac{1}{2}+\delta)$-competitive for any $\delta>0$ even when $q_e=1$ for all $e\in E$.
\end{theorem}
\begin{proof}
Let us consider the following instance, which is well-known for the prophet inequality:
    $U=\{u\}$, $V=\{a,b,c\}$, $T=2$, $p_{a,1}=1, p_{a,2}=0$, $p_{b,1}=0$, $p_{b,2}=\epsilon$, $p_{c,1}=0$, $p_{c,2}=1-\epsilon$, $w_{u,a}=1$, $w_{u,b}=1/\varepsilon$, $w_{u,c}=0$, $C_e$ always takes $T$ ($\forall e\in E$) and $q_e=1$ ($\forall e\in E$).
    
    The optimal offline algorithm chooses $b$ at $t=2$ if $b$ comes, and chooses $a$ at $t=1$ if $c$ comes. The expected profit of the optimal algorithm is $\varepsilon\cdot 1/\varepsilon + (1-\varepsilon)=2-\varepsilon$.
    On the other hand, the expected profit of any algorithm is at most $1$. 
    If the algorithm chooses $a$ at $t=1$, then it obtains $1$. 
    Otherwise, the algorithm obtains $1$ in expectation. 
    Thus, the competitive ratio is at most $1/(2-\varepsilon)$. 
    \end{proof}

\section{Experiments}
In this section, we present experimental results to evaluate the performance of Algorithm~\ref{alg} on synthetic datasets and real-world based ones.

\paragraph{Experiment Environment}
All experiments were conducted on a server equipped with two Xeon E5-2699v3 processors and 768GB of memory, whose OS is CentOS 7.1. 
We implement algorithms by Python 3.7 and use Gurobi Optimizer 8.1.1 to solve LP~(\texttt{Off}).

\paragraph{Instance Setting}
Here we focus on the following four settings (a)--(d).
For the real-world datasets, we consider eleven more settings, which are detailed in Appendix~\ref{sec:experiments-appendix}.
The setting (a) is the KIID model with non-reusable agents~(i.e., $\Pr[C_e =T]=1$ for all $e\in E$) and $\Delta_u<T$ ($\forall u\in U$).
This coincides with that of \citet{Nanda2020}. 
The settings (b)--(d) are the KAD model with reusable agents.
In the setting (b), which is that of \citet{Dickerson2018}, offline vertices always accept assignments~(i.e., $q_e=1$ for all $e\in E$), while they may reject in (c) and (d).
We set $\Delta_u<T$ ($\forall u\in U$) in (c) and $\Delta_u=+\infty$ in (d).
For each setting, we generate one instance and $1000$ input sequences, 
and focus on the average profits and runtime for evaluation.

\paragraph{Baseline Algorithms}
We compare our algorithm with the following baseline algorithms:
(1) {Random} chooses offline vertices uniformly at random. 
(2) Greedy assigns at most $b_v$ available offline vertices $u$ to $v$ in the order of $w_{(u,v)}q_{(u,v)}$.
(3) NAdap* is a simple extension of NAdap proposed by \citet{Nanda2020}.
NAdap* solves LP~(\texttt{Off}) and to assign a set of offline vertices with probability $\lambda^{v,t}_k$.
When $b_v=1$ for all $v$, NAdap* coincides with NAdap in the setting (a), and is also used in experiments in \cite{Dickerson2018}.
In our experiments, we do not evaluate the algorithm by \citet{Dickerson2018} due to its impractical assumption.

\subsection{Synthetic Dataset}
We generate synthetic datasets similarly to \cite{Nanda2020} as follows.
We set $|U|=30$, $|V|=100$, and $T=200$. 
For each $u\in U$ and $v\in V$, an edge $(u,v)$ exists in $E$ with probability $0.1$.
For each $e\in E$, we set $q_e \sim U(0.5,1)$ (uniform distribution) and $w_e \sim U(0,1)$, respectively.
For each $u\in U$ and $e\in E_u$, the distribution of $C_e$ for reusability is defined as a binomial distribution $B(20, \eta_u)$, where $\eta_u\sim U(0,1)$.
For the settings (a) and (c), $\Delta_u$ is drawn uniformly at random from $\{1,2,3\}$ for each $u\in U$.
We set the same $b_v$ for all $v$ from $\{2, 4, 6, 8, 10\}$.

Figure~\ref{fig:synthetic} plots the ratios between average profits and the LP optimal values.
We see that our algorithm (the black solid line) performs the best in almost all cases. 
Moreover, it obtains more than half of the LP optimal value, while Random and Greedy sometimes fail.
NAdap* performs almost the same as us in (a), but is worse in (b)--(d). 

All the algorithms run within less than 1 second for processing online vertices.
Our algorithm needs additionally around 6 seconds on average because of the preprocess of solving LP (\texttt{Off}) and computing the table of $R^d_{u,t}$'s.
For the detailed results, see Appendix~\ref{sec:runtime}.

\begin{figure}[htbp]
    \centering
    \begin{minipage}[b]{0.24\linewidth}
    \includegraphics[keepaspectratio, width=\linewidth]{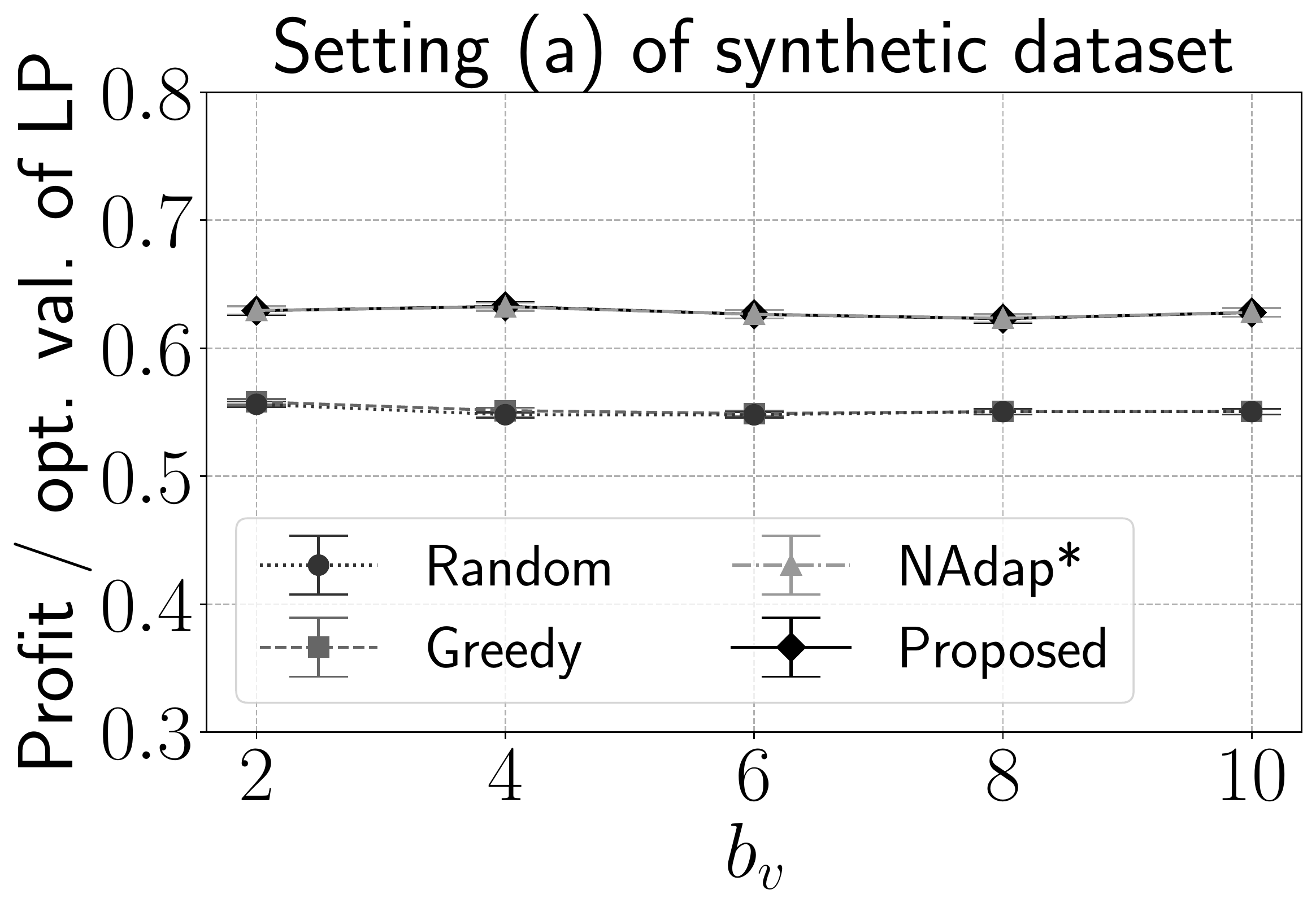}
    \subcaption{KIID without reusability}
    \end{minipage}
    \begin{minipage}[b]{0.24\linewidth}
    \includegraphics[keepaspectratio, width=\linewidth]{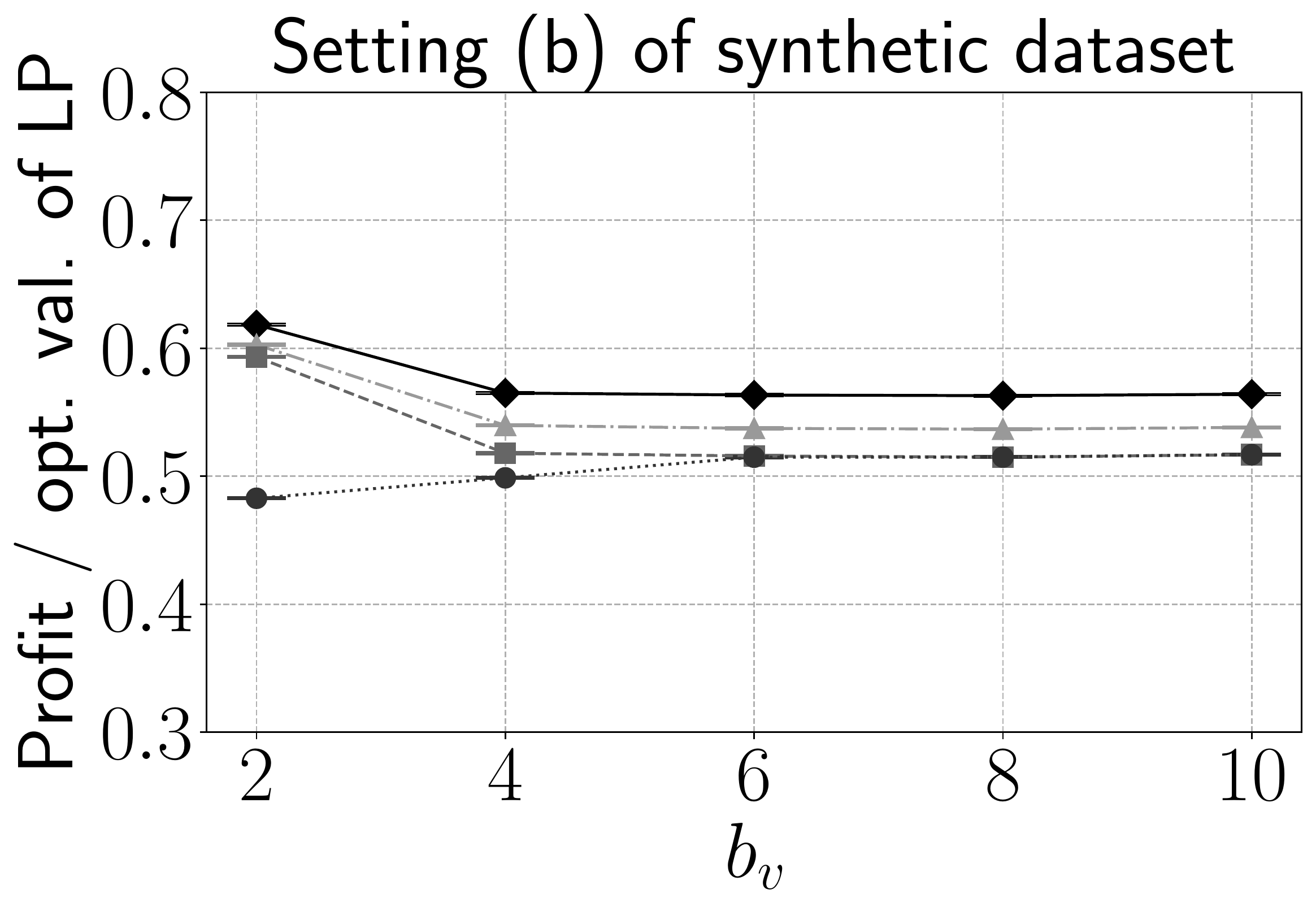}
    \subcaption{KAD without rejection}
    \end{minipage}
    \begin{minipage}[b]{0.24\linewidth}
    \includegraphics[keepaspectratio, width=\linewidth]{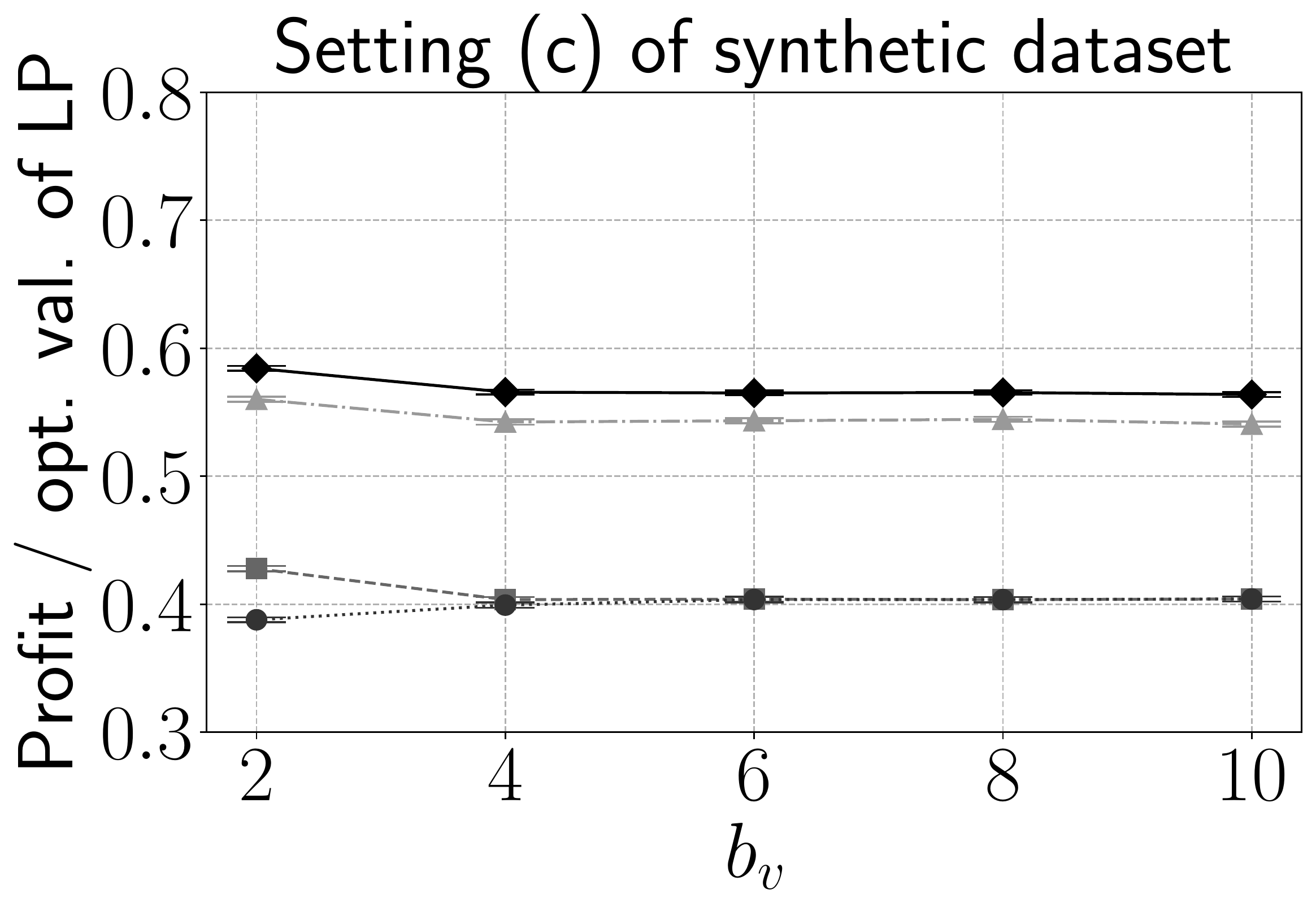}
    \subcaption{KAD with $\Delta_u \leq 3$}
    \end{minipage}
    \begin{minipage}[b]{0.24\linewidth}
    \includegraphics[keepaspectratio, width=\linewidth]{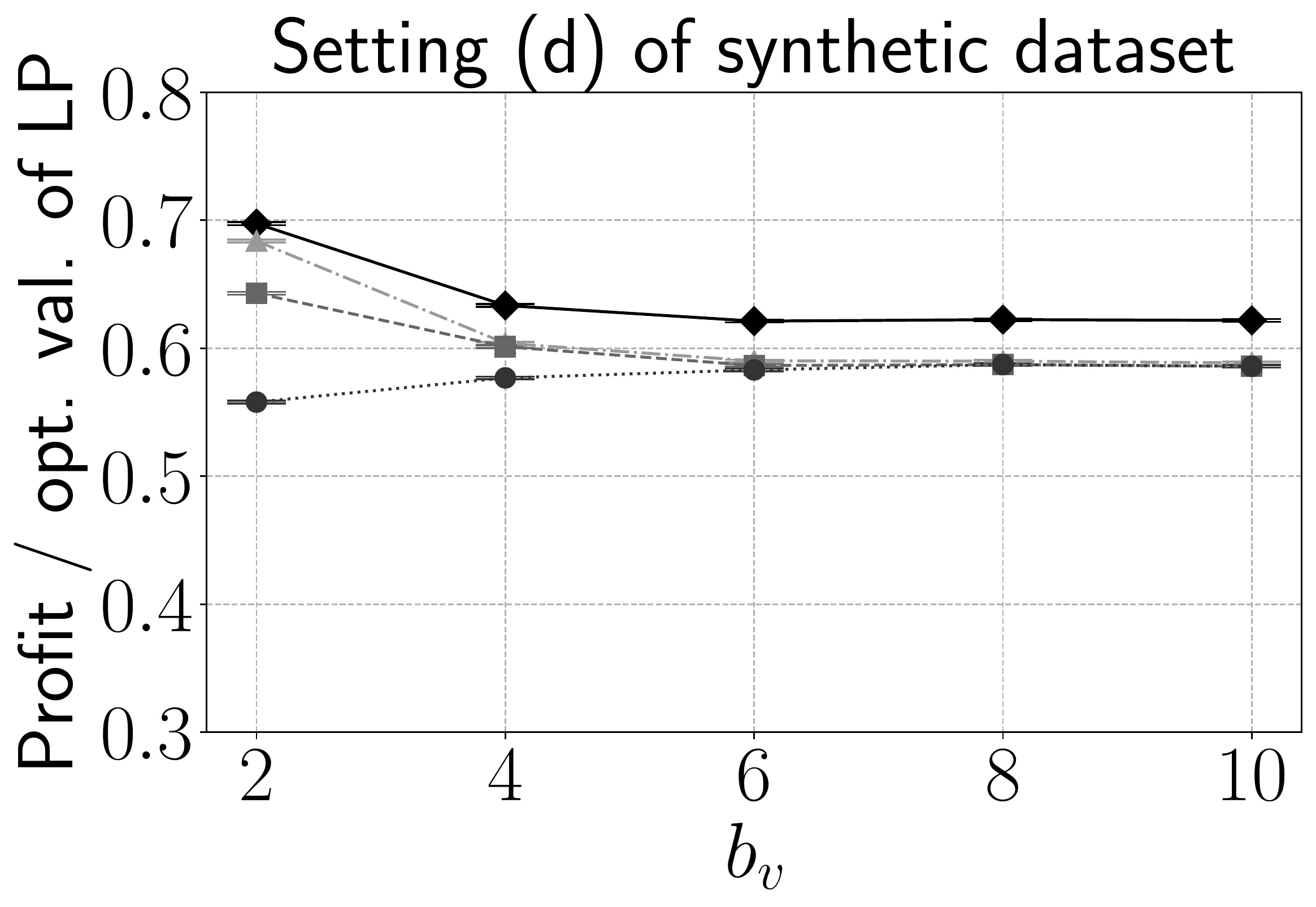}
    \subcaption{KAD with $\Delta_u=+\infty$}
    \end{minipage}
    \caption{The rates of average profits for synthetic datasets.}
    \label{fig:synthetic}
\end{figure}

\subsection{Real-world Dataset}
We evaluate the performance for instances generated from the New York City yellow cabs dataset\footnote{http://www.andresmh.com/nyctaxitrips/},
which is used also in existing work such as~\cite{Dickerson2018,Nanda2020}.
Using the data in January 2013 we set up a bipartite graph with $|U|=30$ and $|V|=100$ and parameters similarly to \cite{Nanda2020}.
The detailed set-up is described in Appendix~\ref{sec:taxi-setup}.
We set $T=100k$ ($k\in [4]$) in setting (a) and $T=288k$ ($k\in [4]$) in others.

Figure~\ref{fig:taxi} shows the ratios between average profits and the LP optimal values.
Note that our algorithm earns more than half of the LP optimal value in any cases. 
In settings (a) and (c), Random and Greedy cannot obtain even half of the LP optimal value for a large $T$.
NAdap* performs as well as ours in settings except (c), in which it performs worse than our algorithm.
Our algorithm may be too careful in (b) and (d), in which there is no need to care about the rejection constraint, and has less performance.
On the other hand,
since offline vertices easily become unavailable in (a) and (c),
our algorithm, which considers a long-term effect of the current assignment, performs the best.

All the algorithms run in less than 1 second for processing even $1152$ online vertices.
The preprocess in our algorithm completes in 400 seconds.
Since the preprocess is done before arrival of online vertices, our algorithm makes decision as fast as others.
For the detailed results, see Appendix~\ref{sec:runtime}.

\begin{figure}[htbp]
    \centering
    \begin{minipage}[b]{0.24\linewidth}
    \includegraphics[keepaspectratio, width=\linewidth]{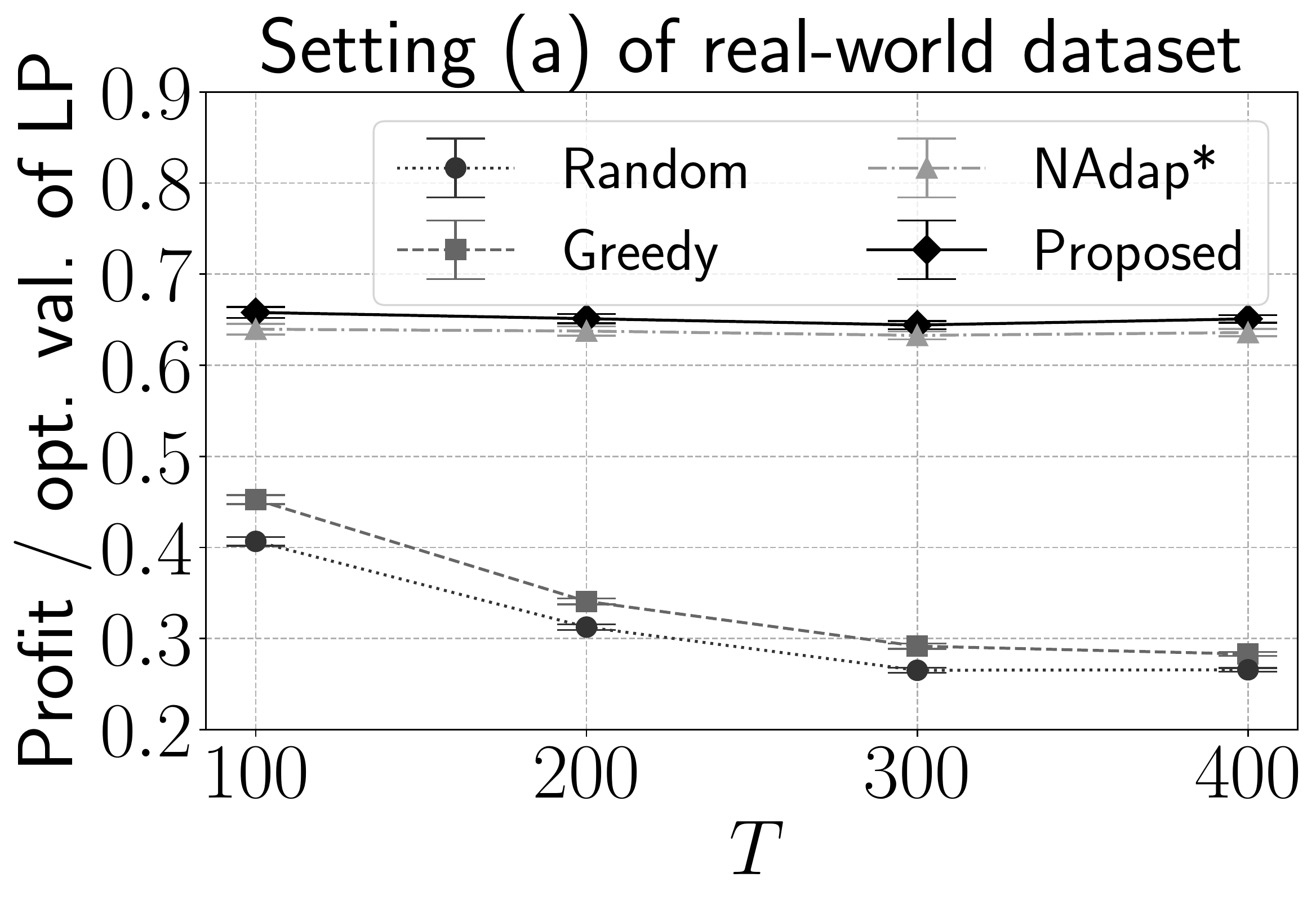}
    \subcaption{KIID without reusability}
    \end{minipage}
    \begin{minipage}[b]{0.24\linewidth}
    \includegraphics[keepaspectratio, width=\linewidth]{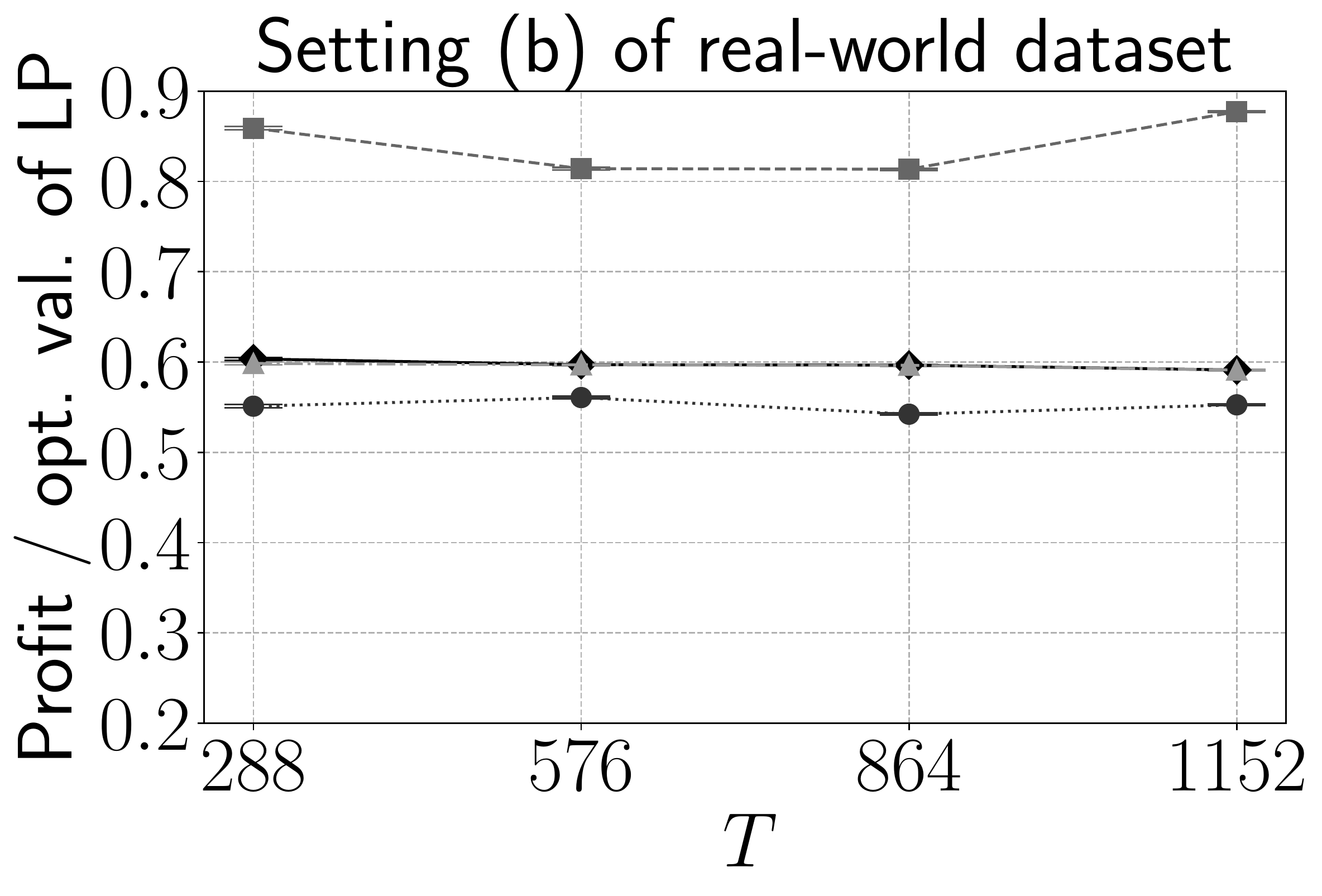}
    \subcaption{KAD without rejection}
    \end{minipage}
    \begin{minipage}[b]{0.24\linewidth}
    \includegraphics[keepaspectratio, width=\linewidth]{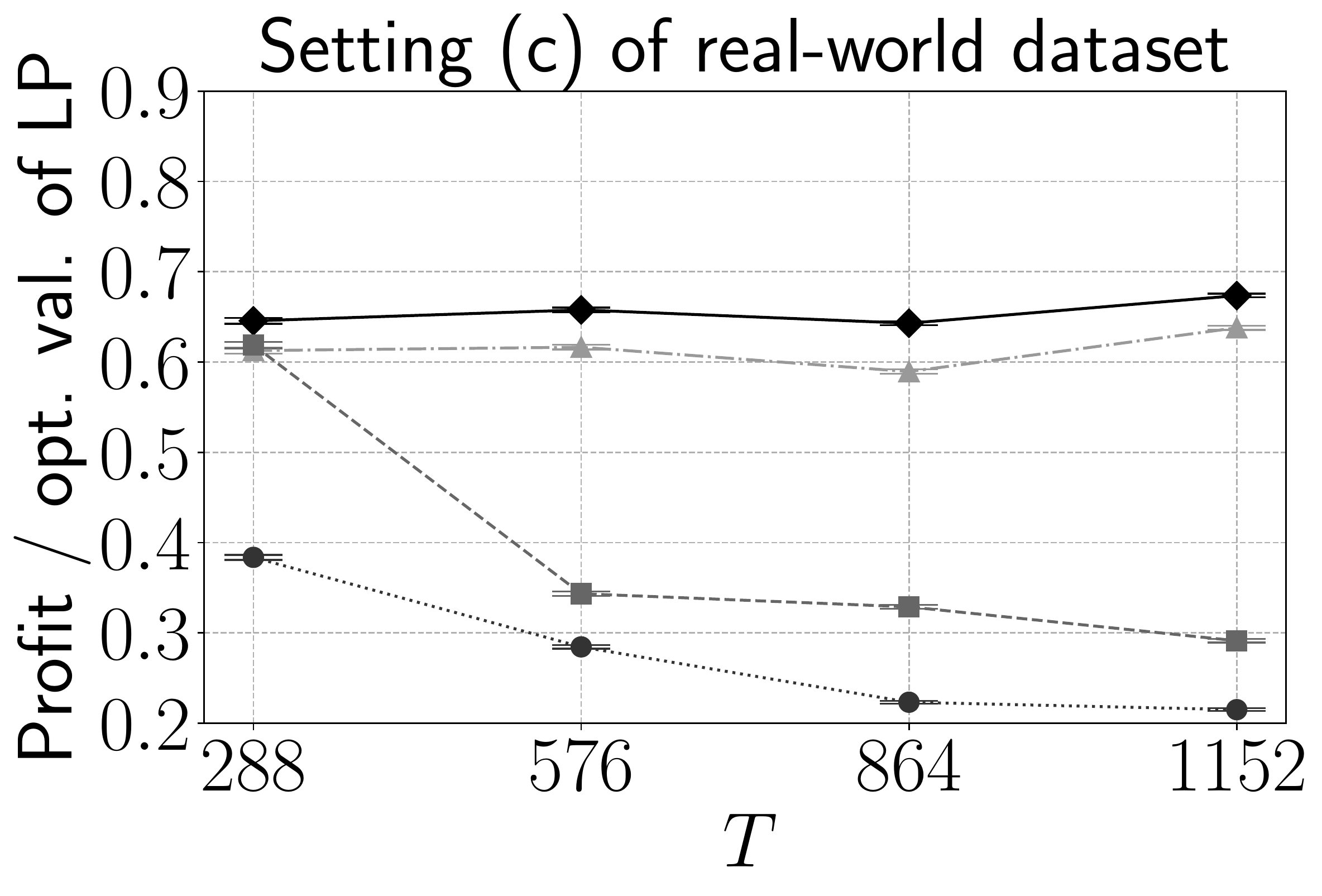}
    \subcaption{KAD with $\Delta_u \leq 3$}
    \end{minipage}
    \begin{minipage}[b]{0.24\linewidth}
    \includegraphics[keepaspectratio, width=\linewidth]{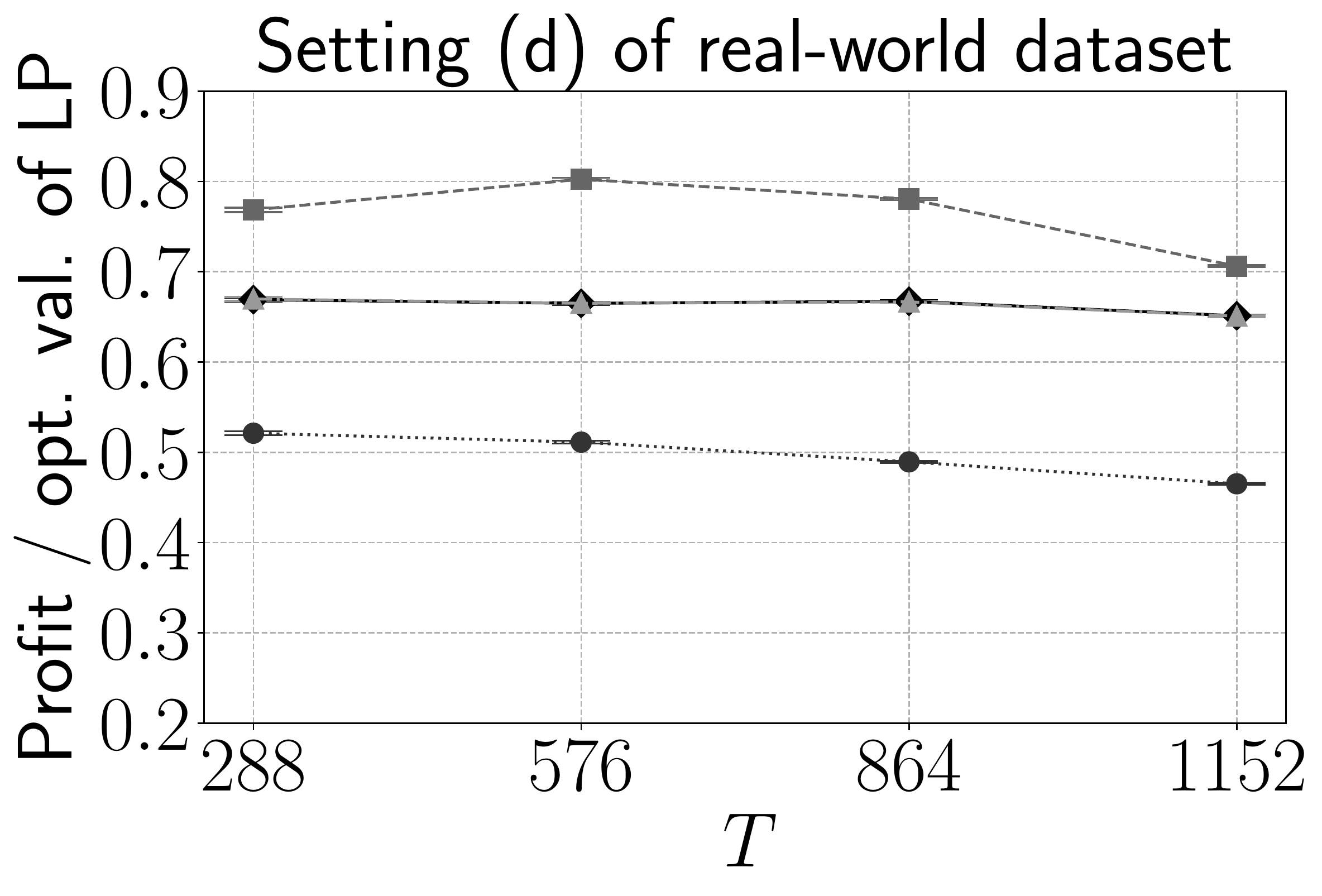}
    \subcaption{KAD with $\Delta_u=+\infty$}
    \end{minipage}
    \caption{The rates of average profits for real-world datasets.}
    \label{fig:taxi}
\end{figure}

\section{Conclusion}
In this paper, we studied the online task assignment problem with reusable resources in the KAD model, which generalizes the online bipartite matching in the KAD model.
Our problem incorporates practical conditions arising in applications such as ridesharing and crowdsourcing.
We proposed an online algorithm that is $1/2$-competitive for the unlimited rejection case, which is tight, and $\Delta/(3\Delta-1)$-competitive for the general case.
Practical usefulness of our algorithm is confirmed by numerical experiments.

For future work,
as solving LP~(\texttt{Off}) is time-consuming,
it would be interesting to obtain a constant competitive ratio without solving the LP.
One direction is to use an approximate solution to the LP in our algorithm.
Another future work is to consider a model that, when a task is rejected, it is allowed to re-assign some other agents.

\section*{Acknowledgments}
HS was supported by JSPS KAKENHI Grant Numbers JP17K12646, JP21K17708, and JP21H03397, Japan.
SI was supported by JST, ACT-I, Grant Number JPMJPR18U5, Japan.
TF was supported by JSPS KAKENHI Grant Numbers JP20H05965, JP21K11759, and JP21H03397, Japan.
NK was supported by JSPS KAKENHI Grant Numbers JP20H05795, JP18H05291, and JP21H03397, Japan.
KK was supported by JSPS KAKENHI Grant Number JP18H05291, Japan.

\bibliographystyle{abbrvnat}
\bibliography{ref}

\clearpage
\appendix
\section{Bad Example of the Simple Extension}\label{sec:simple-is-bad}

We present an instance such that simply assigning offline vertices  according to LP cannot have competitive ratio better than $1/3$.

Let us consider the following instance. 
	Let $U=\{u\}$, $V=\{v_1,v_2\}$ and $E=\{(u,v_1), (u, v_2)\}$. Let also $T=2$ and $\Delta_u=2$.
	At time $t$, only a vertex $v_t$ can arrive, whose probability $p_t$. 
	Thus we identify an edge $(u,v_t)$ with $t$ in this example. 
	The parameters are summarized in Table~\ref{tab:example-bad}.
	Note that this instance is included in the unlimited rejection case.
		\begin{table}[htbp]
		\centering
		\caption{A bad instance for a simple extension.}
		\label{tab:example-bad}
		\begin{tabular}{cccccc}
		\toprule
		    $t$ &  $p_t$ & $q_t$ & $w_t$ & $\Pr[C_t=1]$ & $\Pr[C_t=2]$ \\ \hline
		    1 & 1 & 2/3 & 1 & 0 & 1 \\
		    2 & 1 & 1/3 & $M$ & 0 & 1 \\
		    \bottomrule
		\end{tabular}
		\end{table}
	
	When we assign only $v_1$ (respectively, $v_2$), the expected profit is $\frac{2}{3}$ (respectively, $\frac{M}{3}$).
	When we assign $v_1$ and $v_2$ if possible, the expected profit is $\frac{2}{3}\cdot 1 + \frac{1}{3}\cdot \frac{M}{3} = \frac{M+6}{9}$.
	Thus, the offline optimal value is $\frac{M}{3}$ if $M>3$.
	
	The offline LP (\texttt{Off}) is written specifically as follows:
	\begin{align*}
	    \max \quad & \frac{2}{3}x_1 + \frac{M}{3} x_2 \\
	    \text{s.t.} \quad & \frac{2}{3}x_1 + \frac{1}{3}x_2 \leq 1 \\
	    & x_1+x_2 \leq 2\\
	    & 0 \leq x_1 \leq 1, \ 0 \leq x_2 \leq 1.
	\end{align*}
	This LP has a unique optimal solution $\begin{bmatrix}1 & 1\end{bmatrix}$.
	With this solution, we can consider the algorithm that we choose a vertex $v_t$ with probability $x_t=1$ for $t=1,2$.
	In this case, the expected profit is $\frac{M+6}{9}$.
	Therefore, the ratio to the offline optimal value is $\frac{M+6}{9} \cdot \frac{3}{M} = \frac{M+6}{3M}$, and this is less than $\frac{1}{2}$ if $M>12$.
	Moreover, the ratio approaches $\frac{1}{3}$ when $M$ goes to $+\infty$. 
	
	We note that Algorithm~\ref{alg} attains the expected profit $\frac{M}{3}$, which is the same as the offline optimal value. 

\section{Details of Experiments}\label{sec:experiments-appendix}

We detail experiment setups and present further results.

\paragraph{Instance Settings}
We conduct experiments in fifteen kinds of settings (a)--(o). 
The settings (a)--(d) are applicable to both the synthetic and the real-world datasets, and 
the settings (e)--(o) are applied to only the real-world dataset.
We summarized those settings in Table~\ref{tab:experiment-input}.
In the table, ``NA'' for the rejection constraint means that $q_e = 1$ for all $e\in E$.
For the real-world dataset of taxi trip records in the settings (a), (e), (f) (with a ``Peak-hour'' check-mark), we focus on records in a peak hour of 7--8 PM, similarly to \cite{Nanda2020}.
\begin{table}[htb]
    \centering
    \begin{minipage}[t]{0.6\linewidth}
    \centering
    \caption{Input types of our experiments.}
        \label{tab:experiment-input}
    \scalebox{0.8}{
    \begin{tabular}{ccccc}
        \toprule
        & Online Vertices & Reusable Vertices & \# Rejection & Peak hour \\ \hline
        (a) & KIID &  & $\leq \Delta$ & \checkmark\\
        (b) & KAD & \checkmark & NA & \\
        (c) & KAD & \checkmark & $\leq \Delta$ & \\ 
        (d) & KAD & \checkmark & $+\infty$ & \\ \hline
        \hline
        (e) & KIID &  & NA & \checkmark\\
        (f) & KIID &  & $+\infty$ & \checkmark\\
        (g) & KIID &  & $\leq \Delta$ &\\
        (h) & KIID &  & NA & \\
        (i) & KIID &  & $+\infty$ & \\
        (j) & KIID & \checkmark & $\leq \Delta$ &  \\
        (k) & KIID & \checkmark & NA & \\
        (l) & KIID & \checkmark & $+\infty$ & \\
        (m) & KAD &  & $\leq \Delta$ & \\ 
        (n) & KAD &  & NA & \\
        (o) & KAD &  & $+\infty$ & \\ 
        \bottomrule
    \end{tabular}
    }
    \end{minipage}
\begin{minipage}[t]{0.38\linewidth}
\centering
        \caption{Random seeds.}
    \label{tab:random-seeds}
    \scalebox{0.8}{
    \begin{tabular}{ccc} \toprule
    \multicolumn{2}{c}{Setting} & Seeds \\ \hline
    \multirow{3}{*}{Synthetic}& (a) & 787846414 \\
    & (b) & 3143890026\\ 
    & (c), (d) & 3348747335 \\ \hline
    \multirow{10}{*}{Real-world}& (a), (f) & 1608637542 \\
    & (b) & 3421126067 \\
   & (c), (d) & 4083286876 \\
   & (e) & 249467210 \\
    & (g), (i) & 1914837113 \\
    & (h) & 670094950 \\
    & (j) & 2563451924 \\
    & (k), (l) & 2571218620 \\
    & (m), (o) & 429389014 \\
    & (n) & 669991378 \\ \bottomrule
    \end{tabular}
    }
\end{minipage}
\end{table}

\paragraph{Random Seeds}
We provide the random seeds used to generate input sequences in Table~\ref{tab:random-seeds} .

\subsection{Setup for the Real-World Dataset}\label{sec:taxi-setup}

\paragraph{Dataset Description}
The dataset consists of records of completed taxi trips from Manhattan, Brooklyn, and Queens each month in 2013.
Each record contains an anonymized driver license ID, pick-up/drop-off locations (longtitudes and latitudes) and datetimes, a trip time, and a trip distance. 

\paragraph{Instance Generation from the Real-world Taxi Dataset}
From the data in January 2013, we construct a problem instance as follows.
We set $|U|=30$ and $|V|=100$.
We focus on an area of longitudes from $\ang{-75}$ to $\ang{-73}$ and latitudes from $\ang{40.4}$ to $\ang{40.95}$, and split into grids by step size of $0.05$ as in~\cite{Nanda2020}. 
A grid is called a \emph{location}.
We construct $U$ by sampling $30$ license IDs of drivers uniformly at random.
We assume that each driver is in a location which is chosen randomly from her history of picking up a rider.
The set $V$ of online vertices~(types of riders) is set based on the $100$ most frequently-appeared pairs of pick-up and drop-off locations.
Note that the $100$ types account for $95.1$\% of records in the area. 
For each $u\in U$ and $v\in V$, an edge $e=(u,v)$ exists if the location of driver $u$ is the same as the pick-up location of rider $v$.

The edge profit $w_e$ is set based on a trip distance, as a fare is roughly proportional to a trip distance.
In addition, we reflect a cost of a driver to pick up a rider, similarly to~\cite{Dickerson2018}.
Since a driver can be assigned to riders with a pick-up point in the same location in our experiments, we assume that the pick-up cost mainly depends on the point of a driver.

Thus, we define a pick-up cost $D_u$ for each driver $u\in U$, and 
we choose $D_u$ uniformly at random from $[0, 2.7]$ based on the fact that the mean trip distance is about $2.7$.
Then, we set $w_e = \max(L_v-D_u, 0)$ with $e=(u,v)$, where $L_v$ is a mean trip distance of riders with type $v$.

The acceptance probability $q_e$ is chosen uniformly at random from $[0.5, 1]$.
The distribution of an occupation time $C_e$ is set by a distribution of $C+5$, where $C$ is twice of trip time of riders with type $v$.
This is because we expect that a driver makes a round trip, taking twice of one-way trip time. 
We add $5$ minutes to include the time from the request acceptance to the pick-up. 
We note that the mean of the one-way trip time is around $11$ minutes.
For settings (a) and (c), we set $\Delta_u$ by choosing uniformly at random from $\{1,2,3\}$ for each $u \in U$.

The settings (a), (e), and (f) focus on only trip records within 7--8 PM.
We divide the one-hour period into time slots of at least $100$ and no more than $400$.
On the other hand, in the other settings, we divide the one-day period into $288$ time slots by $5$ minutes as in \cite{Dickerson2018}. 
For the settings (a), (e), and (f), the probability $p_{v,t}$ that an online vertex $v$ arrives is set to be proportional to a frequency of riders with type $v$ during 7--8 PM in the 31 days in January 2013. 
For the KIID settings (g)--(l), $p_{v,t}$ is set to be proportional to a frequency of riders with type $v$ throughout the 31 days.
In the KAD settings, we set $p_{v,t}$ to be proportional to a frequency of riders with type $v$ in the $t$-th time slot in the 31 days, for each time slot $t$. 

\subsection{Further Results for the Real-World Dataset}\label{sec:taxi-further}
In this subsection, we present further results for the real-world dataset and add observations.
Figure~\ref{fig:taxi-additional} shows the ratios between the average profits and the LP optimal value of LP (\texttt{Off}) in settings (e)--(o).

\begin{figure}[ht]
    \centering
    \begin{subfigure}[b]{0.24\linewidth}
    \addtocounter{subfigure}{4}
    \includegraphics[keepaspectratio, width=\linewidth]{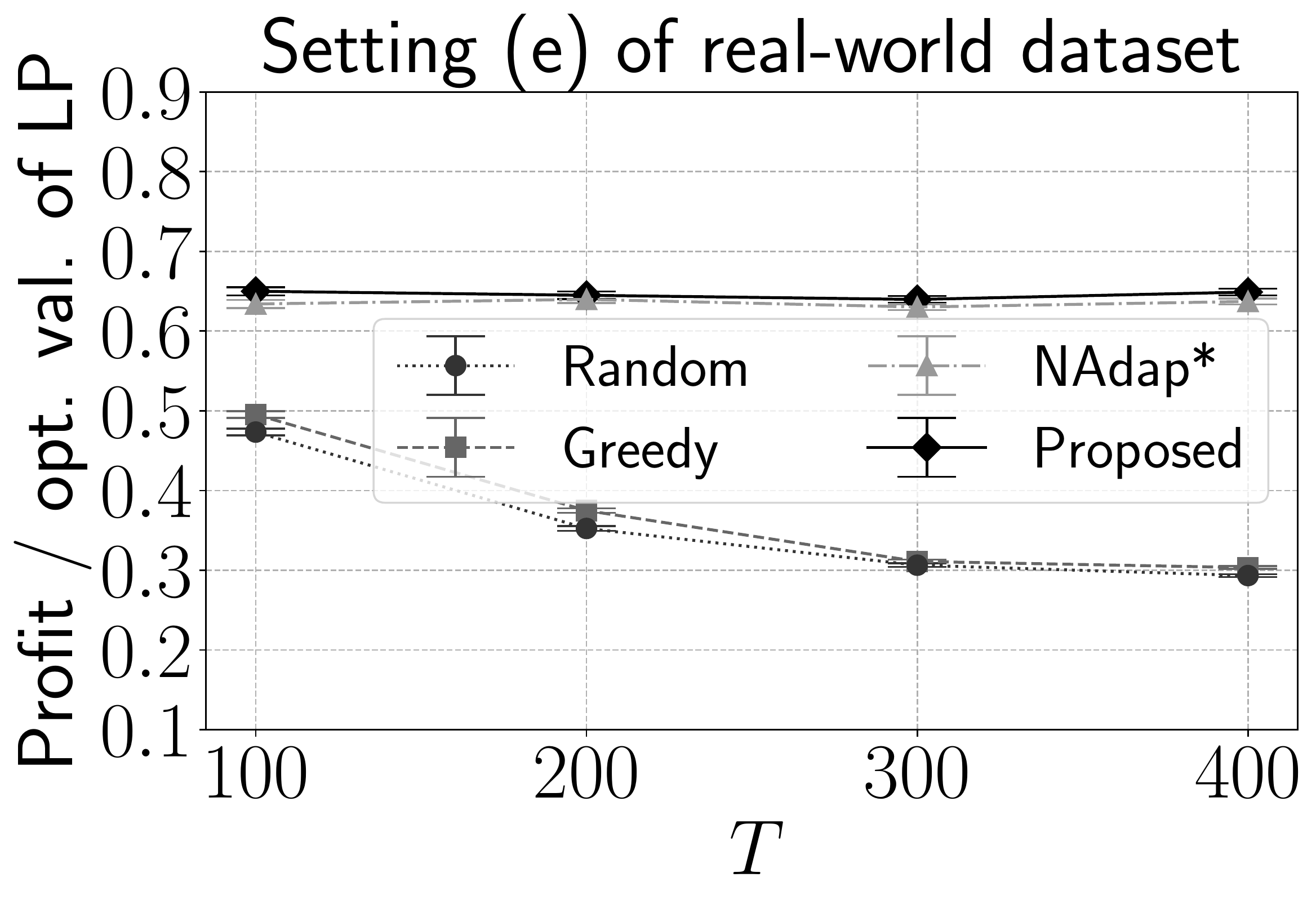}
    \subcaption{KIID for the peak hour without rejections}
    \end{subfigure}
    \begin{minipage}[b]{0.24\linewidth}
    \includegraphics[keepaspectratio, width=\linewidth]{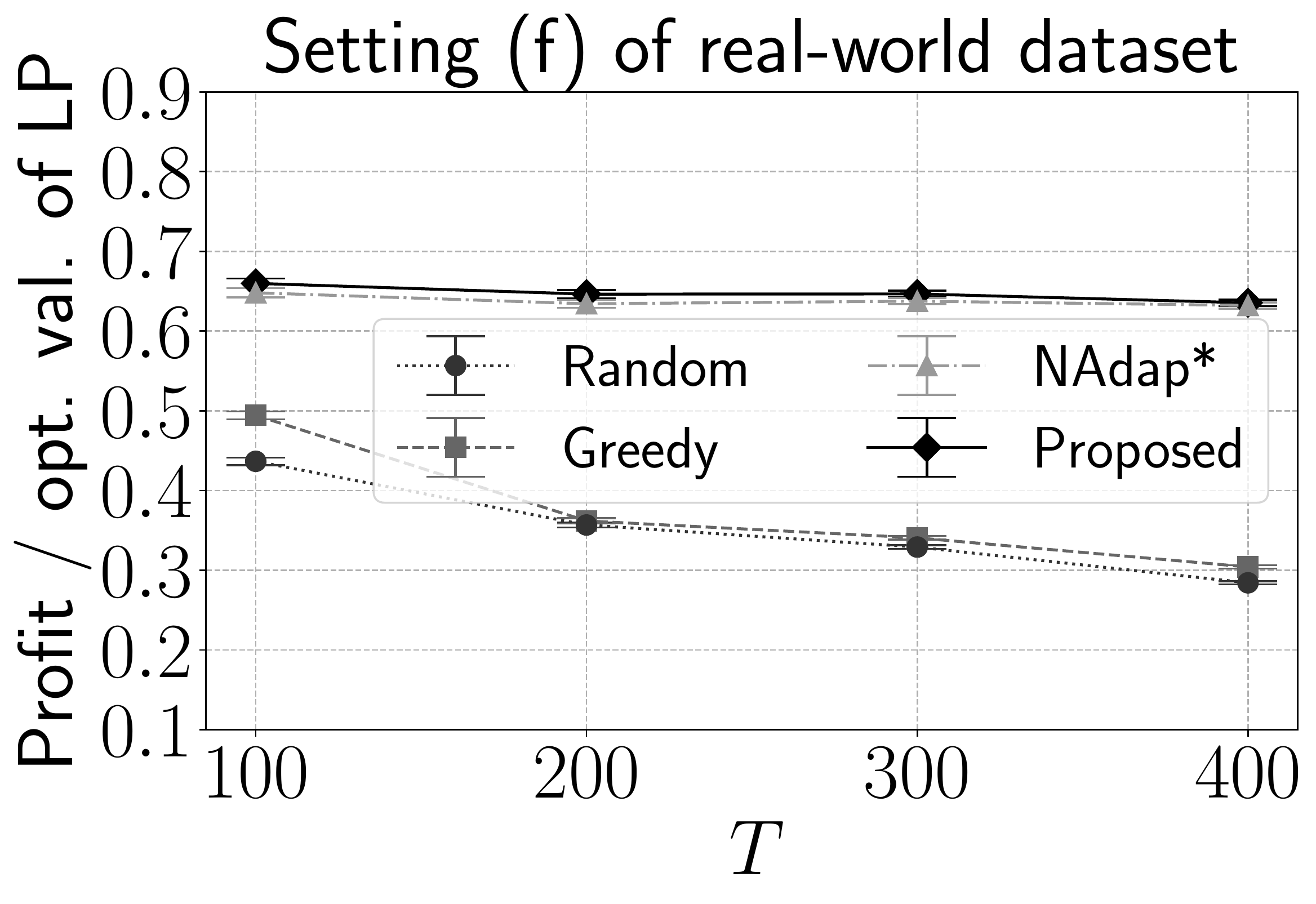}
    \subcaption{KIID for the peak hour with $\Delta_u=+\infty$}
    \end{minipage}
    \begin{minipage}[b]{0.24\linewidth}
    \includegraphics[keepaspectratio, width=\linewidth]{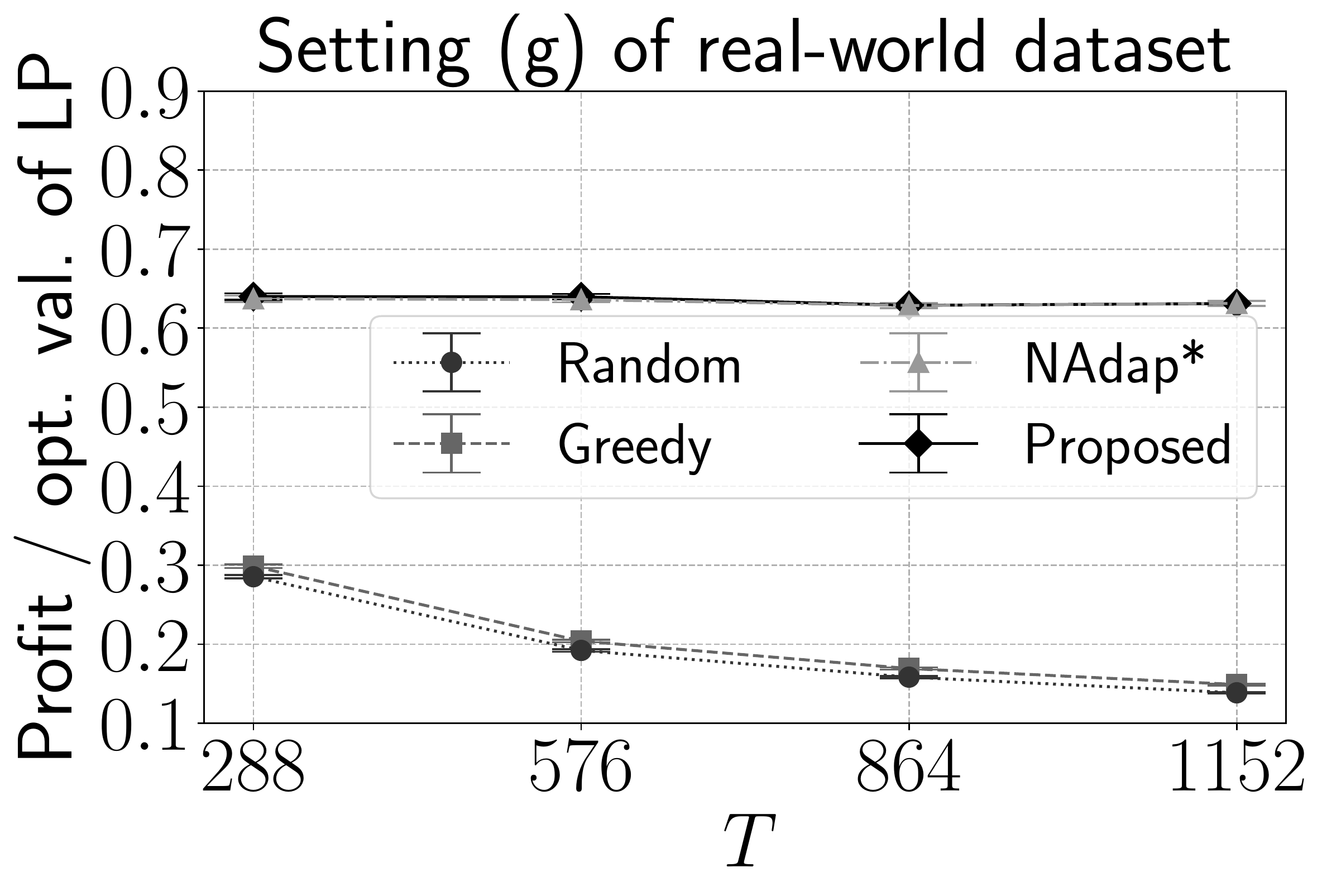}
    \subcaption{KIID without reusability \\ ($\Delta_u \leq 3$)}
    \end{minipage}
    \begin{minipage}[b]{0.24\linewidth}
    \includegraphics[keepaspectratio, width=\linewidth]{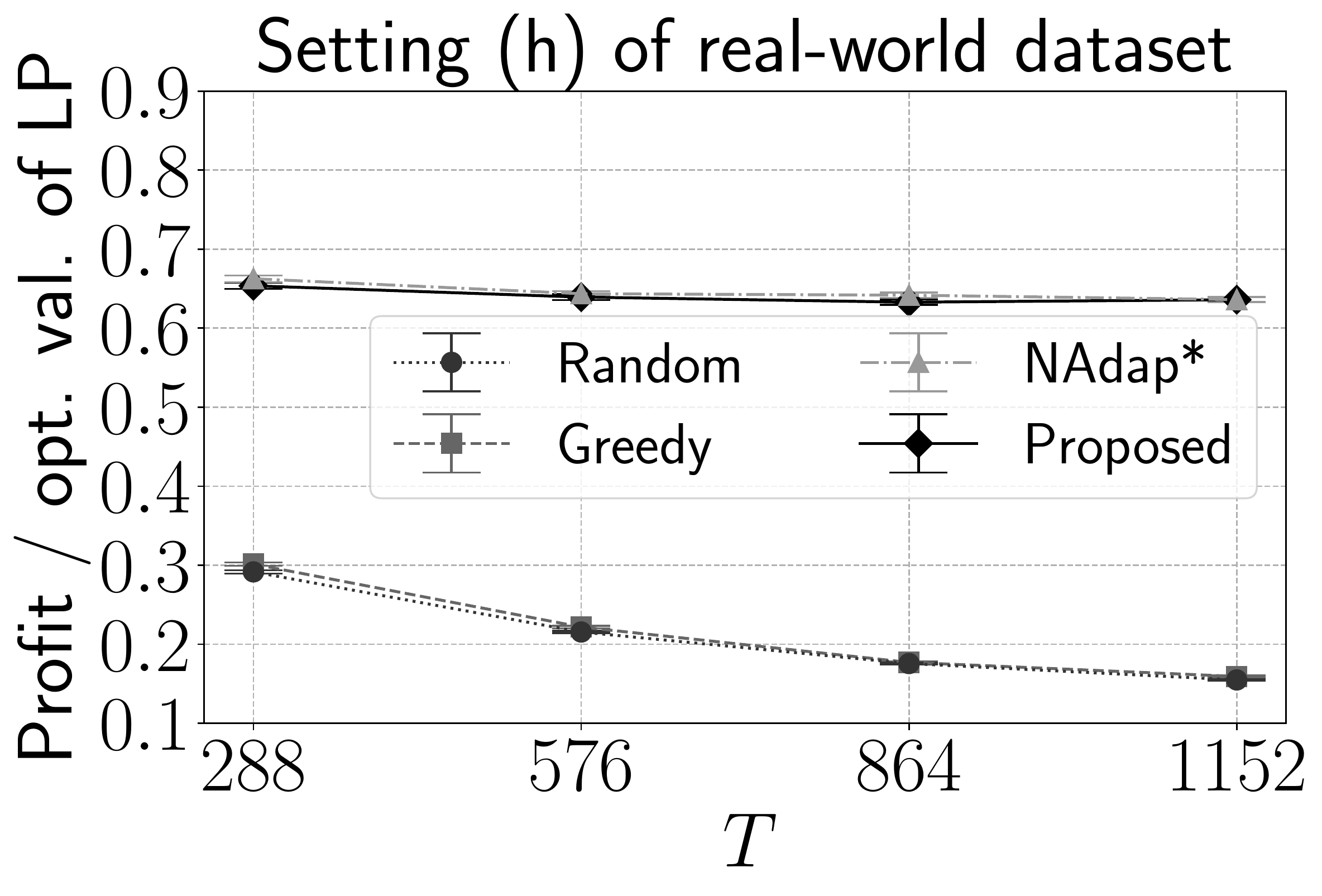}
    \subcaption{KIID without reusability and rejections}
    \end{minipage}
    
    \medskip
    
    \begin{minipage}[b]{0.24\linewidth}
    \includegraphics[keepaspectratio, width=\linewidth]{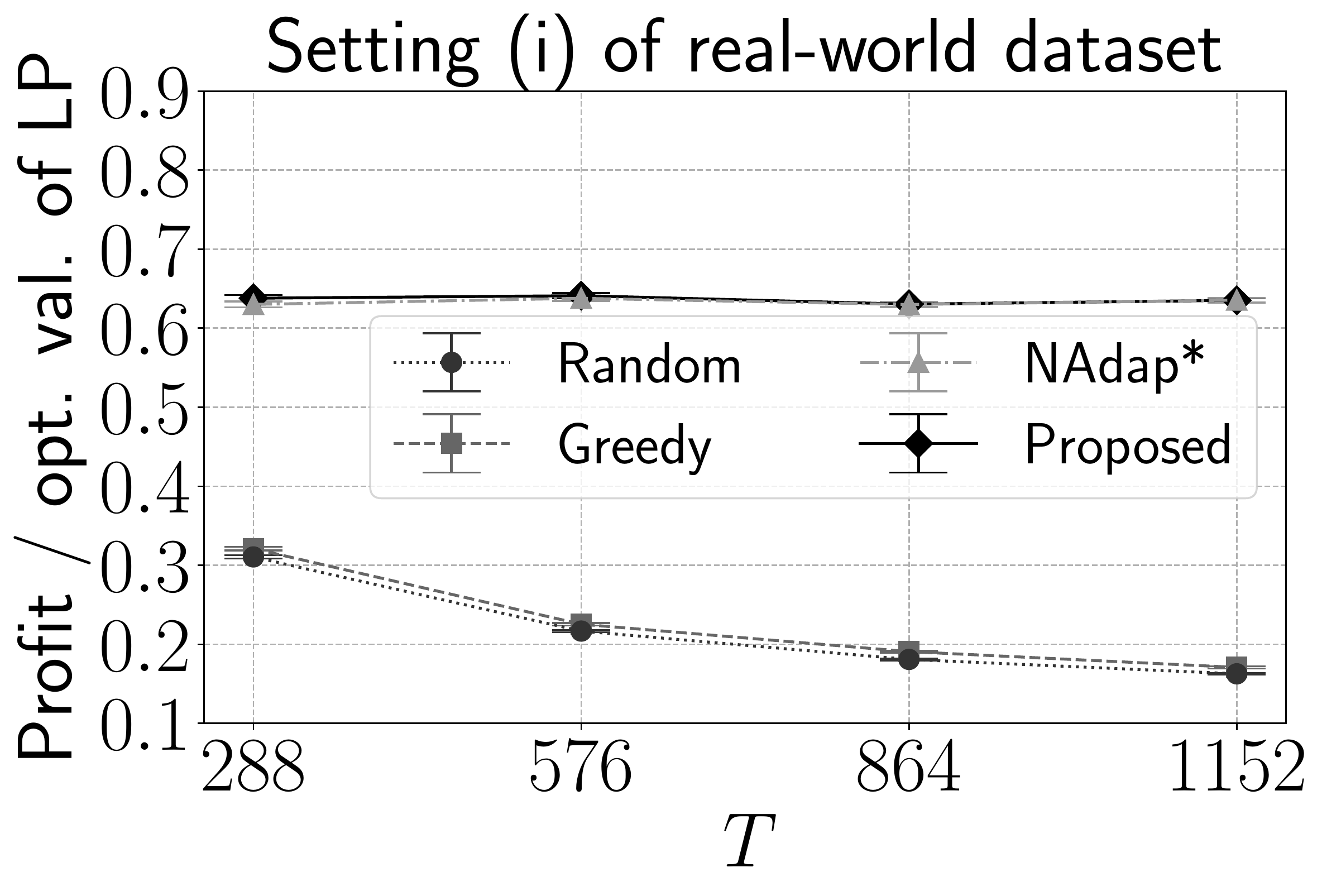}
    \subcaption{KIID without reusability \\ ($\Delta_u=+\infty$)}
    \end{minipage}
    \begin{minipage}[b]{0.24\linewidth}
    \includegraphics[keepaspectratio, width=\linewidth]{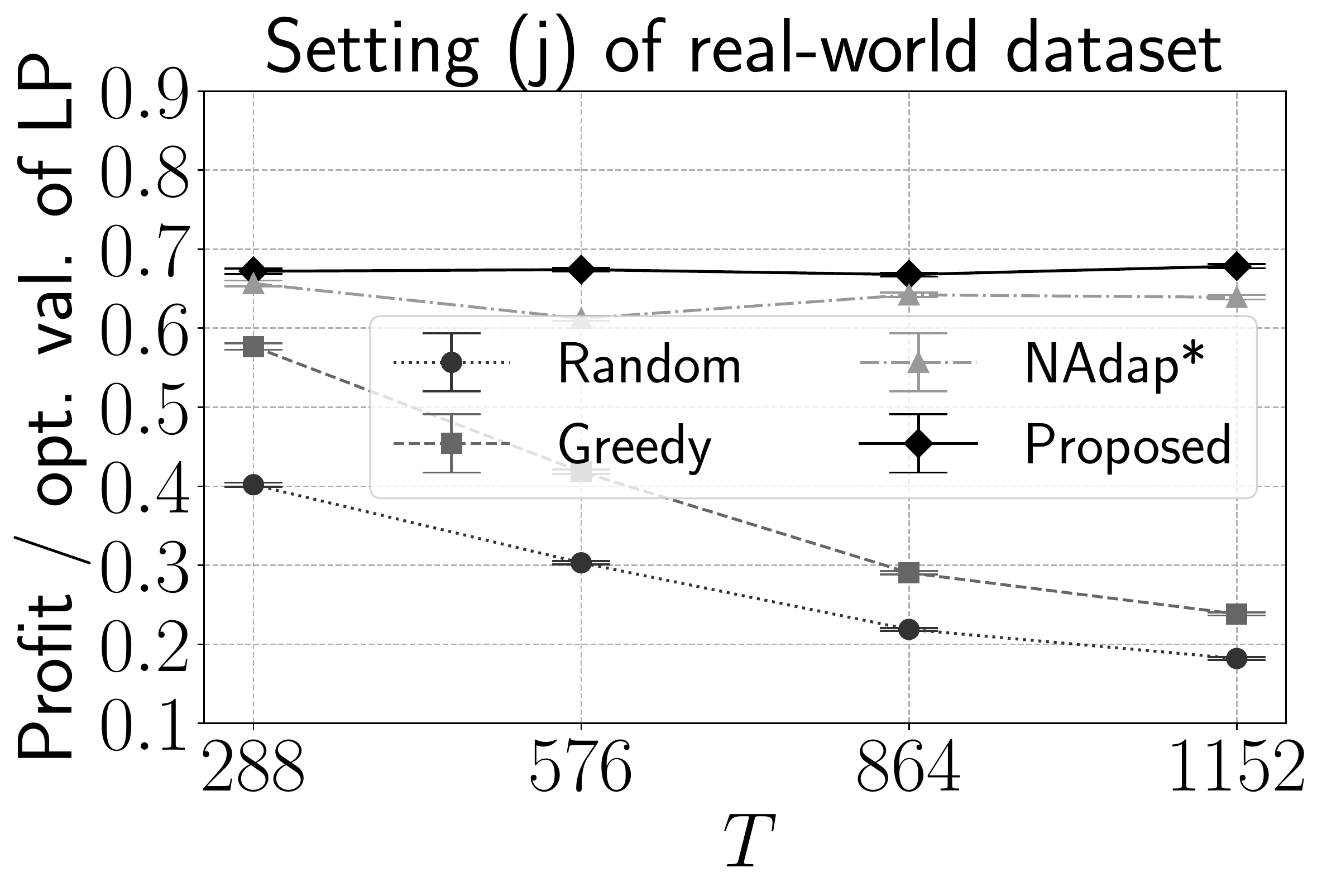}
    \subcaption{KIID with $\Delta_u\leq 3$ \\ \phantom{a}}
    \end{minipage}
    \begin{minipage}[b]{0.24\linewidth}
    \includegraphics[keepaspectratio, width=\linewidth]{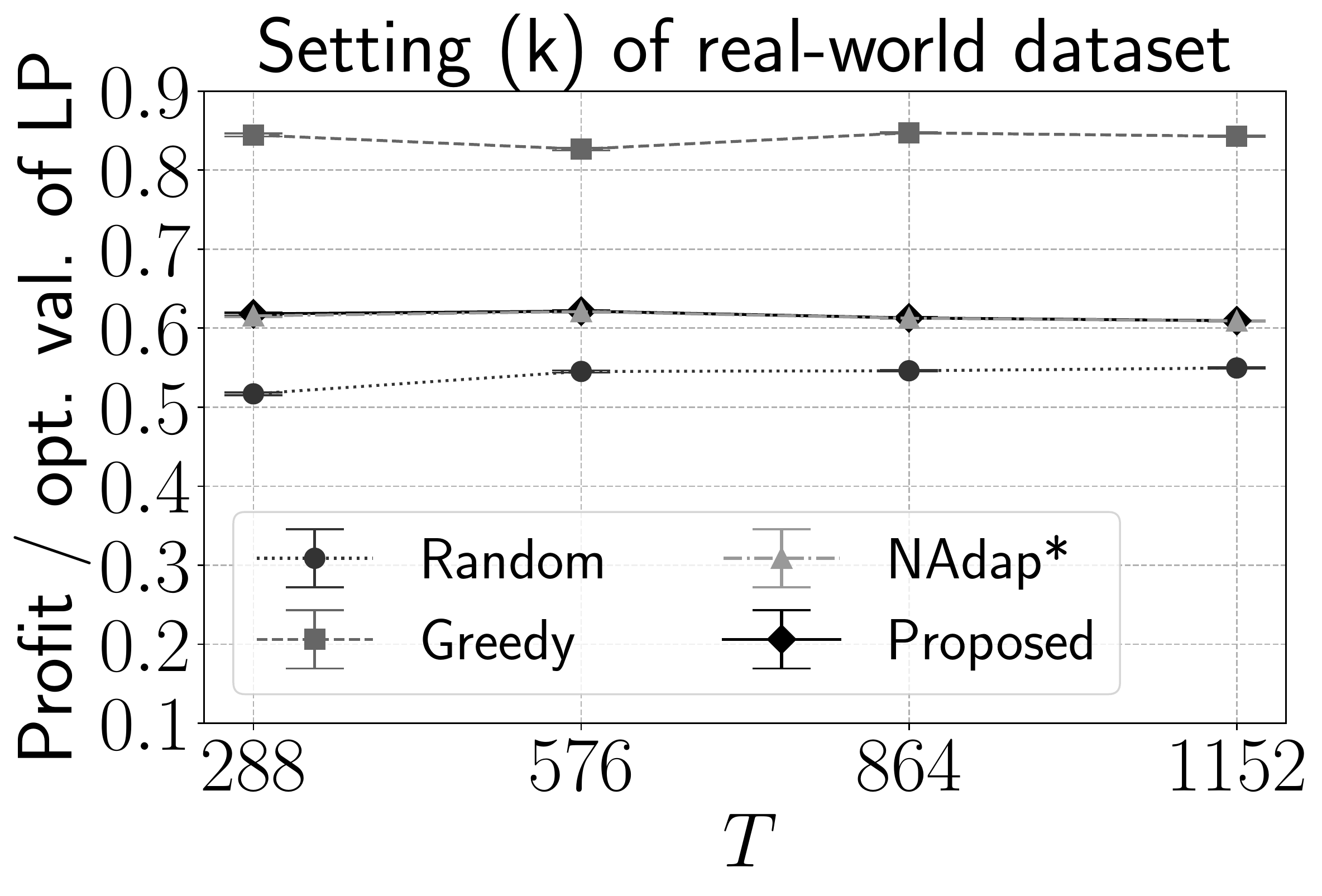}
    \subcaption{KIID without rejections \\ \phantom{a}}
    \end{minipage}
    \begin{minipage}[b]{0.24\linewidth}
    \includegraphics[keepaspectratio, width=\linewidth]{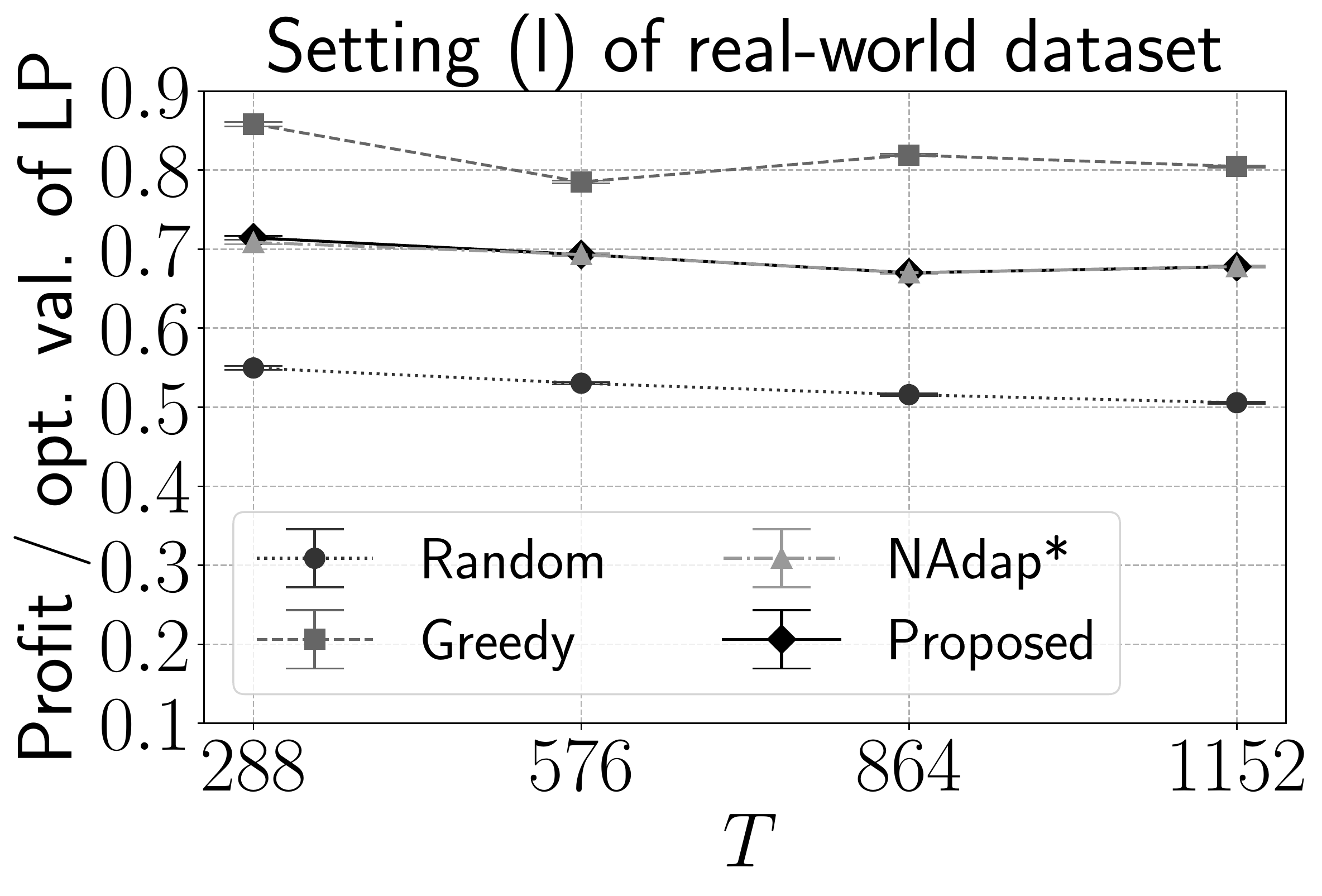}
    \subcaption{KIID with $\Delta_u=+\infty$ \\ \phantom{a}}
    \end{minipage}
    
    \medskip
    
    \begin{minipage}[b]{0.24\linewidth}
    \includegraphics[keepaspectratio, width=\linewidth]{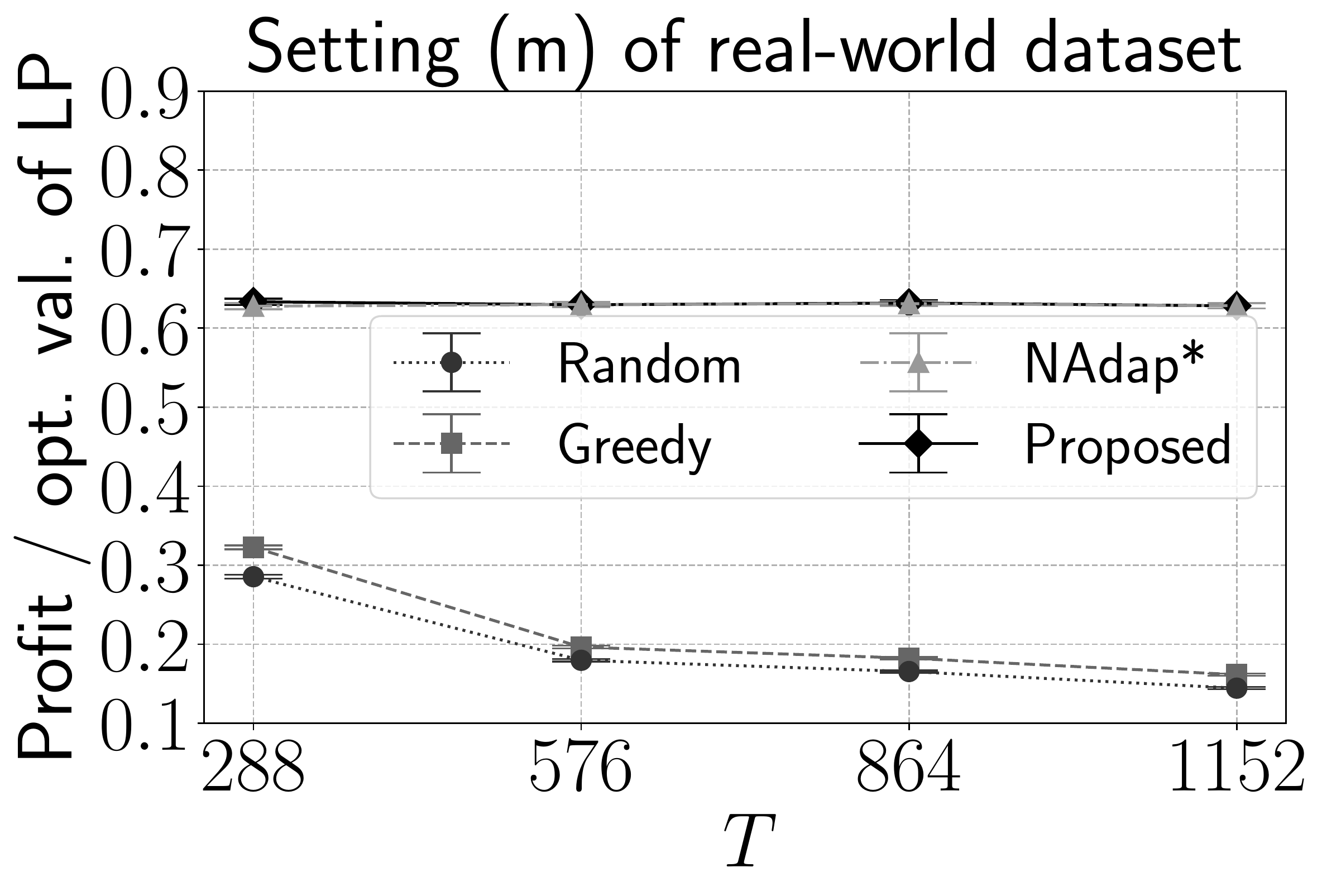}
    \subcaption{KAD without reusability ($\Delta_u \leq 3$)}
    \end{minipage}
    \begin{minipage}[b]{0.24\linewidth}
    \includegraphics[keepaspectratio, width=\linewidth]{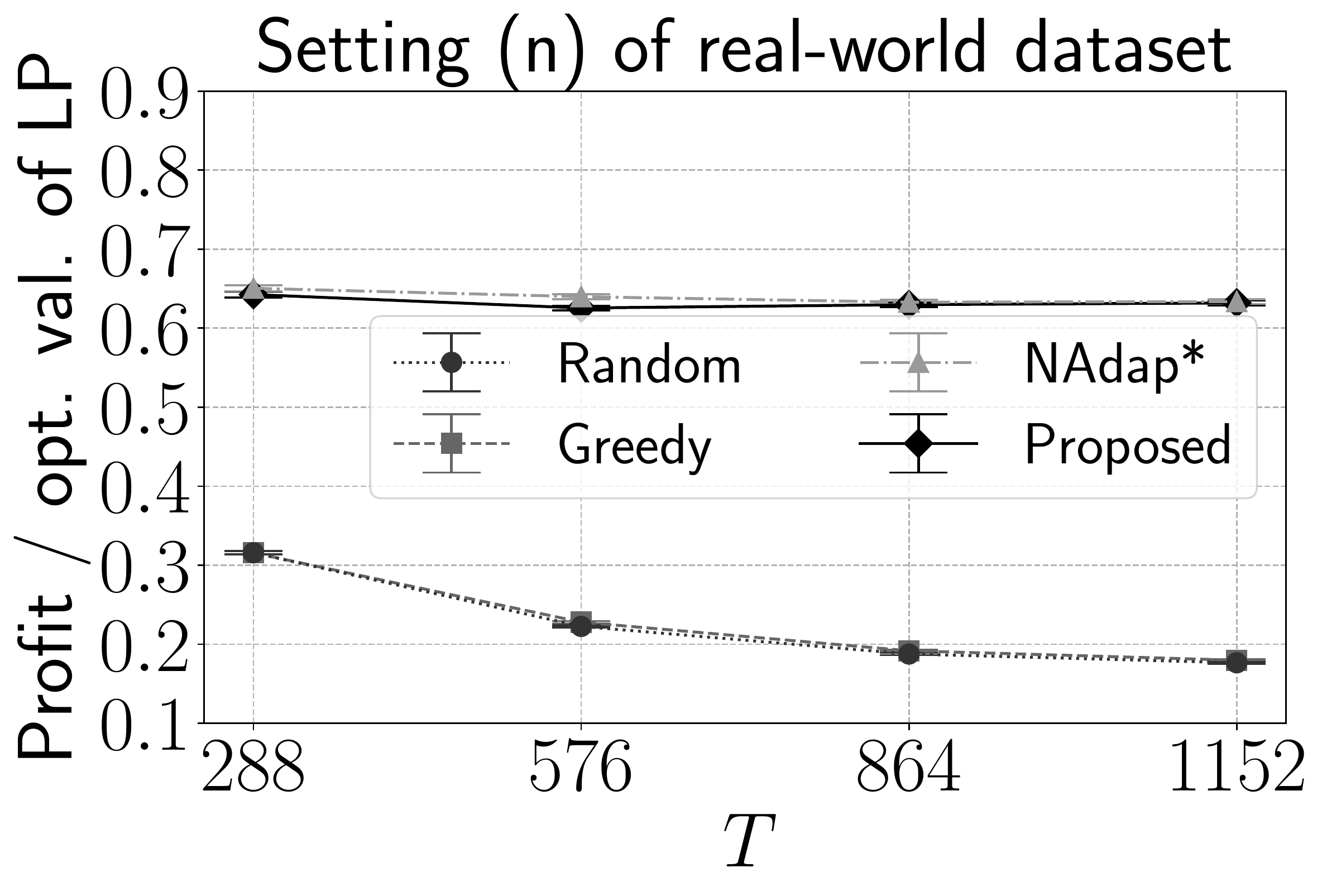}
    \subcaption{KAD without reusability, rejections}
    \end{minipage}
    \begin{minipage}[b]{0.24\linewidth}
    \includegraphics[keepaspectratio, width=\linewidth]{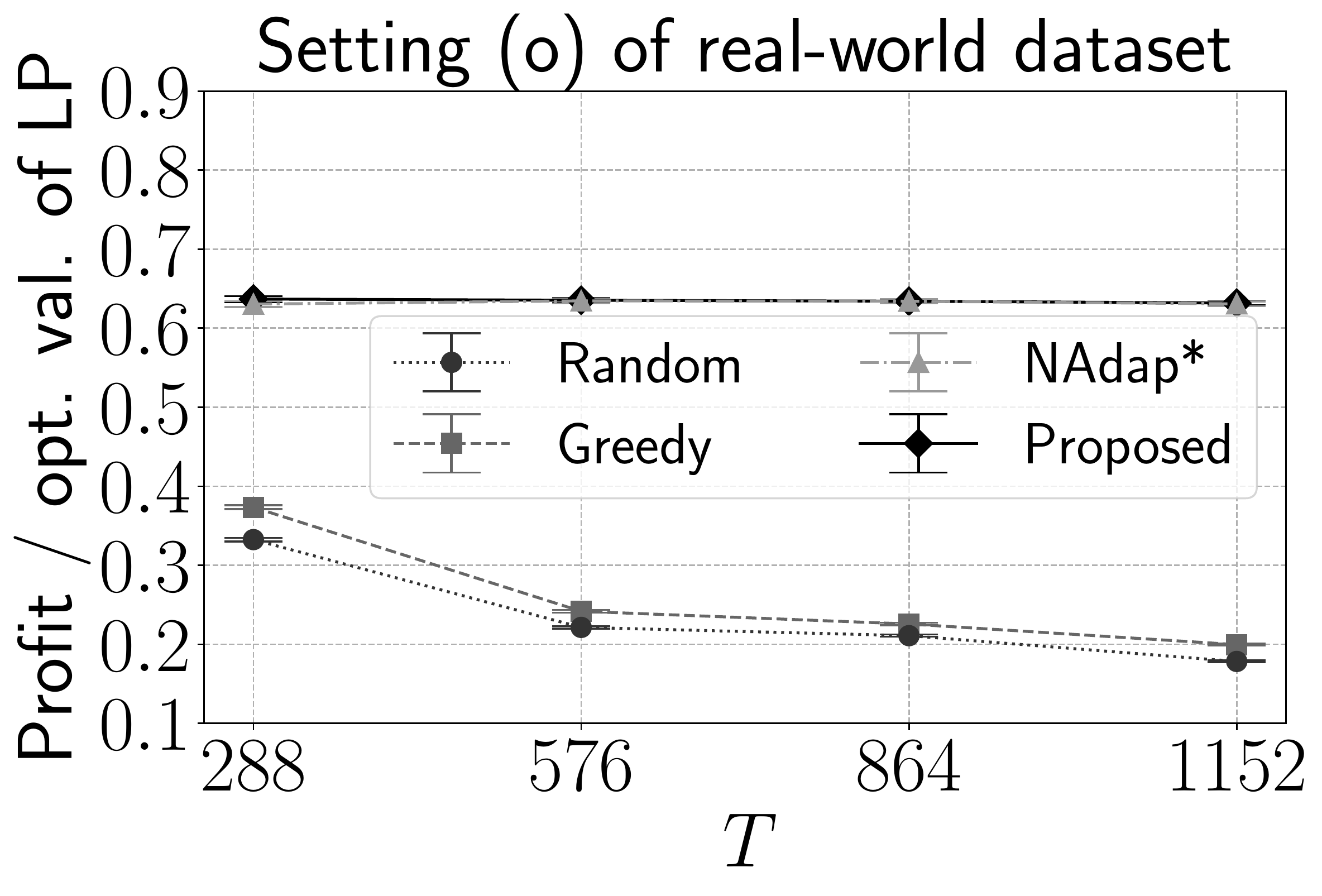}
    \subcaption{KAD without reusability ($\Delta_u=+\infty$)}
    \end{minipage}
    
    \caption{The rates of average profits for real-world datasets in settings (e)--(o).}
    \label{fig:taxi-additional}
\end{figure}

The results indicate that the proposed algorithm has a stably high performance.
The performance of NAdap* is mostly similar to ours, and the difference is clear in (j), where offline vertices easily become unavailable. 
We also observe that by a similar reason to (b) and (d), our algorithm has less performance in (k) and (l), where offline vertices rarely become unavailable.
On the other hand, the proposed algorithm always obtains more than half of the LP optimal value, while Random and Greedy sometimes fail. 
This is prominent when $T$ is larger in settings except (k) and (l).

These results together with those in (a)--(d) imply that the proposed algorithm performs well, and even better than baseline algorithms in most cases.

\subsection{Detailed Results on Runtime}\label{sec:runtime}

We present the detail of average runtime.
Table~\ref{tab:syn-time}--\ref{tab:taxi-time-kad} show the average runtimes of algorithms and the preprocess.
In the tables, ``\veryfast'' indicates that the runtime is shorter than $0.005$ seconds.
We also denote Random by ``\Random'', Greedy by ``\Greedy'', NAdap* by ``\NAdap'', the proposed algorithm by ``\Proposed'', and the preprocess of the proposed algorithm by ``Prep'' in the tables.
Recall that the preprocess of the proposed algorithm consists of solving LP (\texttt{Off}) and computing $R^d_{u,t}$'s. 
We exclude the time needed for buiding a model in the LP solver because this is very specific to the programming language and the solver. 
Note that in the settings (a), (e), and (f) for the real-world dataset, NAdap* uses a simpler LP~\eqref{eq:offlineLP-kiid} with $b_v=1$ ($v\in V$) for the purpose of reproducing NAdap~\cite{Nanda2020}.

We mention that all the algorithms process online vertices in less than 1 second even for $T=1125$. 
Our algorithm takes much time due to the preprocess. 
In fact, our algorithm takes $1347$ seconds in the setting (m) with $T=1152$, which is the most time-consuming case for us.
We see that our algorithm in total is practically fast enough, and it is as fast as baseline algorithms during processing online vertices.

\begin{table}[hb]
    \centering
    \begin{minipage}[t]{0.45\linewidth}
    \centering
    \caption{Average runtime[sec] for synthetic datasets.}
    \label{tab:syn-time}
    \scalebox{0.75}{
    \begin{tabular}{ccccccc} \toprule
        && \multicolumn{5}{c}{ $b_v$} \\
        && $2$ & $4$ & $6$ & $8$& $10$\\ \hline
        \multirow{5}{*}{(a)} 
        & \Random & \veryfast & \veryfast & \veryfast & \veryfast & \veryfast \\
        & \Greedy & \veryfast & \veryfast &\veryfast &\veryfast & \veryfast \\
        & \NAdap & 1.26&1.23&1.04&1.05&0.94 \\
        & \Proposed & 2.73&2.29&2.65&2.32&2.3 \\
        \cline{2-7}
        & Prep & 2.72&2.29&2.64&2.32&2.3 \\ \hline
        \multirow{5}{*}{(b)}
        & \Random & 0.01&0.01&0.01&0.01&0.01 \\
        & \Greedy & 0.01&0.01&0.01&0.01&0.01 \\
        & \NAdap & 1.96&1.71&1.33&1.35&1.24 \\
        & \Proposed & 5.83&5.4&5.06&5.15&5.19 \\
        \cline{2-7}
        & Prep & 5.82&5.38&5.04&5.13&5.17 \\ \hline
        \multirow{5}{*}{(c)} 
        & \Random & 0.01&0.01&0.01&0.01&0.01 \\
        & \Greedy & 0.01&0.01&\veryfast&0.01&0.01 \\
        & \NAdap & 1.09&0.87&0.94&0.83&0.77 \\
        & \Proposed & 6.36&6.17&5.99&6.23&5.82 \\
        \cline{2-7}
        & Prep & 6.35&6.16&5.98&6.22&5.81 \\ \hline
        \multirow{5}{*}{(d)} 
        & \Random & 0.01&0.01&0.01&0.01&0.01 \\
        & \Greedy & 0.01&0.01&0.01&0.01&0.01 \\
        & \NAdap & 0.97&0.81&0.88&0.94&0.9 \\
        & \Proposed & 5.19&4.79&4.46&4.4&4.41 \\
        \cline{2-7}
        & Prep & 5.17&4.77&4.45&4.39&4.39 \\ \bottomrule
    \end{tabular}
    }
\end{minipage}
\hspace{1em}
\begin{minipage}[t]{0.45\linewidth}
    \centering
     \caption{Average runtime[sec] for real-world datasets: the KIID settings with peak-hour data.}
     \label{tab:taxi-time-peak}
    \scalebox{0.75}{
    \begin{tabular}{cccccc} \toprule
        & & \multicolumn{4}{c}{$T$} \\
        &  & $100$ & $200$ & $300$ & $400$\\ \hline
        \multirow{5}{*}{(a)} 
        & \Random & \veryfast&0.01&0.01&0.01 \\
        & \Greedy & \veryfast&0.01&0.01&0.02 \\
        & \NAdap & 0.01&0.02&0.01&0.01 \\
        & \Proposed & 6.21&22.59&50.58&100.4 \\
        \cline{2-6}
        & Prep & 6.2&22.59&50.57&100.39 \\ \hline
        \multirow{5}{*}{(e)}
        & \Random & \veryfast&0.01&0.01&0.01 \\
        & \Greedy & \veryfast&0.01&0.01&0.02 \\
        & \NAdap & 0.01&0.02&0.02&0.01 \\
        & \Proposed & 5.45&19.83&41.5&79.76 \\
        \cline{2-6}
        & Prep & 5.44&19.82&41.49&79.76 \\ \hline
        \multirow{5}{*}{(f)} 
        & \Random & \veryfast&0.01&0.01&0.01 \\
        & \Greedy & \veryfast&0.01&0.01&0.02 \\
        & \NAdap & 0.02&0.01&0.01&0.01 \\
        & \Proposed & 5.15&16.71&40.12&79.05 \\
        \cline{2-6}
        & Prep & 5.14&16.7&40.11&79.04 \\ \bottomrule
    \end{tabular}
    }
     \end{minipage}
     
\bigskip

    \begin{minipage}[t]{0.45\linewidth}
    \centering
    \caption{Average runtime[sec] for real-world datasets: the KIID settings with the whole data.}
     \label{tab:taxi-time-kiid}
    \scalebox{0.75}{
    \begin{tabular}{cccccc} \toprule
        & & \multicolumn{4}{c}{$T$} \\
        & & $288$ & $576$ & $864$ & $1152$\\ \hline
        \multirow{5}{*}{(g)} 
        & \Random & 0.01&0.02&0.03&0.04 \\
        & \Greedy & 0.01&0.03&0.03&0.04 \\
        & \NAdap & 26.68&132.71&410.4&1011.11 \\
        & \Proposed & 49.52&213.44&635.45&1346.18 \\
        \cline{2-6}
        & Prep & 49.52&213.43&635.44&1346.17 \\ \hline
        \multirow{5}{*}{(h)} 
        & \Random & 0.01&0.02&0.03&0.05 \\
        & \Greedy & 0.01&0.02&0.03&0.06 \\
        & \NAdap & 22.36&136.64&494.58&1076.51 \\
        & \Proposed & 36.54&191.17&536.78&586.82 \\
        \cline{2-6}
        & Prep & 36.54&191.16&536.77&586.8 \\ \hline
        \multirow{5}{*}{(i)} 
        & \Random & 0.01&0.02&0.03&0.04 \\
        & \Greedy & 0.01&0.02&0.03&0.04 \\
        & \NAdap & 28.06&133.03&400.41&1051.34 \\
        & \Proposed & 41.12&188.09&525.87&1247.75 \\
        \cline{2-6}
        & Prep & 41.12&188.08&525.86&1247.74 \\ \hline
        \multirow{5}{*}{(j)} 
        & \Random & 0.01&0.02&0.03&0.04 \\
        & \Greedy & 0.01&0.03&0.04&0.05 \\
        & \NAdap & 11.17&17.14&28.19&29.44 \\
        & \Proposed & 32.62&105.73&246.09&371.35 \\
        \cline{2-6}
        & Prep & 32.6&105.7&246.05&371.31 \\ \hline
        \multirow{5}{*}{(k)} 
        & \Random & 0.01&0.02&0.03&0.04 \\
        & \Greedy & 0.02&0.03&0.05&0.07 \\
        & \NAdap & 10.42&23.81&38.05&75.76 \\
        & \Proposed & 22.4&75.42&157.52&279.69 \\
        \cline{2-6}
        & Prep & 22.38&75.38&157.46&279.62 \\ \hline
        \multirow{5}{*}{(l)} 
        & \Random & 0.01&0.02&0.03&0.04 \\
        & \Greedy & 0.02&0.03&0.05&0.07 \\
        & \NAdap & 3.22&7.01&10.92&14.76 \\
        & \Proposed & 16.82&61.92&135.52&238.17 \\
        \cline{2-6}
        & Prep & 16.8&61.89&135.46&238.09 \\ \bottomrule
    \end{tabular}
    }
\end{minipage}
\hspace{1em}
\begin{minipage}[t]{0.45\linewidth}
\centering
\caption{Average runtime[sec] for real-world datasets: the KAD settings.}
     \label{tab:taxi-time-kad}
\scalebox{0.75}{
    \begin{tabular}{cccccc} \toprule
        & & \multicolumn{4}{c}{$T$} \\
        & & $288$ & $576$ & $864$ & $1152$\\ \hline
                \multirow{5}{*}{(b)} 
        & \Random & 0.01&0.02&0.03&0.04 \\
        & \Greedy & 0.02&0.03&0.05&0.07 \\
        & \NAdap & 11.93&31.62&60.05&79.41 \\
        & \Proposed & 25.63&82.88&178.34&293.51 \\
        \cline{2-6}
        & Prep & 25.61&82.84&178.28&293.43 \\ \hline
         \multirow{5}{*}{(c)} 
        & \Random & 0.01&0.02&0.03&0.05 \\
        & \Greedy & 0.01&0.03&0.04&0.06 \\
        & \NAdap & 8.12&16.32&53.67&54.24 \\
        & \Proposed & 29.27&103.51&223.56&448.22 \\
        \cline{2-6}
        & Prep & 29.25&103.48&223.52&448.15 \\ \hline
        \multirow{5}{*}{(d)} 
        & \Random & 0.01&0.02&0.03&0.04 \\
        & \Greedy & 0.02&0.03&0.05&0.06 \\
        & \NAdap & 4.27&9.19&11.46&17.57 \\
        & \Proposed & 17.95&65.47&137.05&242.53 \\
        \cline{2-6}
        & Prep & 17.93&65.43&136.99&242.46 \\ \hline
        \multirow{5}{*}{(m)}
        & \Random & 0.01&0.02&0.03&0.04 \\
        & \Greedy & 0.01&0.02&0.03&0.04 \\
        & \NAdap & 28.53&128.22&387.42&996.4 \\
        & \Proposed & 49.35&201.62&586.96&1347.24 \\
        \cline{2-6}
        & Prep & 49.35&201.61&586.95&1347.23 \\ \hline
        \multirow{5}{*}{(n)}
        & \Random & 0.01&0.02&0.03&0.04 \\
        & \Greedy & 0.01&0.02&0.03&0.06 \\
        & \NAdap & 26.71&140.02&398.93&1017.74 \\
        & \Proposed & 36.84&186.18&305.01&543.98 \\
        \cline{2-6}
        & Prep & 36.84&186.17&305.0&543.96 \\ \hline
        \multirow{5}{*}{(o)} 
        & \Random & 0.01&0.02&0.03&0.04 \\
        & \Greedy & 0.01&0.02&0.03&0.05 \\
        & \NAdap & 26.15&132.93&414.53&1036.29 \\
        & \Proposed & 44.34&191.5&536.48&1237.22 \\
        \cline{2-6}
        & Prep & 44.34&191.49&536.47&1237.21 \\ \bottomrule
    \end{tabular}
    }
     \end{minipage}
\end{table}

\end{document}